\documentclass[11pt,letterpaper]{article}
\usepackage[utf8]{inputenc}
\usepackage[T1]{fontenc}
\usepackage{microtype}

\usepackage[left=1in, right=1in, top=1in, bottom=1in]{geometry}%
\usepackage{subcaption}
\usepackage{nicefrac}
\usepackage{amsmath}
\usepackage{amsfonts}
\usepackage{amssymb}
\usepackage{amsthm}
\usepackage{amstext}
\usepackage{thmtools}
\usepackage{relsize}

\usepackage{csquotes}
\usepackage{xcolor}

\usepackage{mathtools}
\usepackage{xspace}
\usepackage{enumerate}

\usepackage{soul}

\usepackage{paralist}
\usepackage{bm}

\usepackage{url}
\usepackage[colorlinks]{hyperref}
\hypersetup{citecolor=green!60!black,linkcolor=red!70!black,urlcolor=blue!70!black}
\usepackage[capitalize]{cleveref}
\definecolor{buy}{RGB}{182,18,49}
\definecolor{sell}{HTML}{82e057}
\definecolor{lightblue}{HTML}{77CDFF}

\usepackage{tikz}
\usepackage{pgfplots}
\usetikzlibrary{decorations.pathreplacing}
\usetikzlibrary{calc,tikzmark}
\usetikzlibrary{positioning,shapes,arrows} 
\usetikzlibrary{matrix}
\tikzset{
	vx/.style={circle,fill,black,minimum size=7pt,inner sep=0pt,outer sep=0pt},
	comp/.style={circle,line width=2pt,draw=black,minimum size=0.6cm,inner sep=0pt, outer sep=1pt},
	pac/.style={circle,line width=0pt,minimum size=0.2cm,inner sep=0pt, outer sep=-2pt},
	edge/.style={draw=black,line width=1.7pt},
	ledge/.style={draw=black!20!white,line width=1.4pt,dashed},
	mypath/.style={draw=black!20!white,line width=1.4pt},
}
\usepackage[inline]{enumitem}
\usepackage{thm-restate}
\usepackage{algorithm}
\usepackage[noend]{algpseudocode}
\usepackage[normalem]{ulem} 

\usepackage[
 backend=biber, 
 style=alphabetic, 
 citetracker, 
 hyperref=auto,
 maxcitenames=5, 
 sortcites,      
 sorting=nyt,
 maxbibnames=12, 
 giveninits,     
 date=year,      
 isbn=false,     
 url=false,      
 doi=false,      
 eprint=false,   
]{biblatex}
\newbibmacro{string+doiurlisbn}[1]{%
  \iffieldundef{doi}{%
    \iffieldundef{url}{%
      #1
    }{%
      \href{\thefield{url}}{#1}%
    }%
  }{%
    \href{http://dx.doi.org/\thefield{doi}}{#1}%
  }%
}

\DeclareFieldFormat%
[article,inbook,incollection,inproceedings,patent,thesis,unpublished]
  {title}{\usebibmacro{string+doiurlisbn}{#1}} 

\theoremstyle{plain}
\newtheorem{theorem}{Theorem}
\newtheorem{invariant}{Invariant}
\newtheorem{lemma}{Lemma}
\newtheorem{claim}{Claim}
\newtheorem{definition}{Definition}

\newtheorem{observation}{Observation}
\newtheorem{fact}{Fact}
\newtheorem{remark}{Remark}
\newtheorem{corollary}{Corollary}

\newcommand{\eps}{\varepsilon}

\newcommand{\ALG}{\mathsf{ALG}}

\newcommand{\RED}{\mathsf{RED}}
\newcommand{\OPT}{\mathsf{OPT}}
\newcommand{\opt}{\mathsf{opt}}

\newcommand{\APX}{\mathsf{APX}}

\newcommand{\cost}{\mathsf{cost}}
\newcommand{\credit}{\mathsf{cr}}
\newcommand{\poly}{\mathsf{poly}}

\usepackage[colorinlistoftodos,textsize=tiny,textwidth=2cm,color=red!25!white,obeyFinal]{todonotes}

\newcounter{casei}
\newcounter{caseii}[casei]
\newcounter{caseiii}[caseii]

\newenvironment{caseanalysis}{\setcounter{casei}{0}}{}

\newcommand{\case}[1]{%
\refstepcounter{casei}%
\noindent\textbf{\boldmath(\alph{casei}) #1}\unboldmath%
\phantomsection%
}
\newcommand{\subcase}[1]{%
\refstepcounter{caseii}%
\noindent\textbf{\boldmath(\alph{casei}.\arabic{caseii}) #1}\unboldmath%
\phantomsection%
}
\newcommand{\subsubcase}[1]{%
\refstepcounter{caseiii}%
\noindent\textbf{\boldmath(\alph{casei}.\arabic{caseii}.\roman{caseiii}) #1}\unboldmath%
\phantomsection%
}

\def\DEBUG{true}

\ifdefined\DEBUG

    \newcommand{\fel}[1]{\textcolor{teal}{#1}}
    \newcommand{\moh}[1]{\textcolor{purple}{#1}}

    \newcommand{\fabr}[1]{\todo{\textcolor{red}{$\bullet$ #1}}}
    \newcommand{\afr}[1]{\todo{\textcolor{olive}{$\bullet$ #1}}}
    \newcommand{\mig}[1]{\todo{\textcolor{blue}{$\bullet$ #1}}}
    \newcommand{\alex}[1]{\todo[backgroundcolor=white]{\textcolor{orange}{$\bullet$ #1}}}
    
    \newcommand{\felix}[1]{\todo[backgroundcolor=white]{\textcolor{teal}{$\bullet$ #1}}}
    \newcommand{\mohi}[1]{\todo{\textcolor{purple}{$\bullet$ #1}}}
\else
    
    \newcommand{\fabr}[1]{}
    
    \newcommand{\mig}[1]{}
    \newcommand{\alex}[1]{}
    \newcommand{\felix}[1]{}
    \newcommand{\fel}[1]{}

    \newcommand{\afr}[1]{}
    \newcommand{\moh}[1][#1]
    \newcommand{\mohi}[1]{}
    
\fi

\title{A $5/4$-Approximation for Two-Edge Connectivity}

\author{Miguel Bosch-Calvo\thanks{\texttt{\{miguel.boschcalvo,fabrizio\}@idsia.ch}, IDSIA, USI-SUPSI, Lugano, Switzerland. Partially supported by the SNF Grant 200021\_200731 / 1.} \and Mohit Garg\thanks{\texttt{mohitgarg@iisc.ac.in}, Indian Institute of Science, Bengaluru, India. Supported by a fellowship from the Walmart Center for Tech Excellence at IISc (CSR Grant WMGT-23-
0001).} \and Fabrizio Grandoni\footnotemark[1] \and Felix Hommelsheim\thanks{\texttt{\{fhommels,linderal\}@uni-bremen.de}, Faculty of Mathematics and Computer Science, University of Bremen, Germany.} \and Afrouz Jabal Ameli\thanks{\texttt{a.jabal.ameli@tue.nl}, TU Eindhoven, Eindhoven, The Netherlands.} \and Alexander Lindermayr\footnotemark[3]}
\date{}

\begin{document}

\maketitle
\begin{abstract}
The $2$-Edge-Connected Spanning Subgraph problem (2ECSS) is among the most basic survivable network design problems: given an undirected and unweighted graph, the task is to find a spanning subgraph with the minimum number of edges that is 2-edge-connected (i.e., it remains connected after the removal of any single edge). 2ECSS is an NP-hard problem that has been extensively studied in the context of approximation algorithms. The best known approximation ratio for 2ECSS prior to this work was $1.3+\eps$, for any constant $\eps>0$ [Garg, Grandoni, Jabal-Ameli'23; Kobayashi, Noguchi'23]. In this paper, we present a $5/4$-approximation algorithm. Our algorithm is also faster for small values of $\eps$: its running time is $n^{O(1)}$ instead of $n^{O(1/\eps)}$.  
\end{abstract}

\thispagestyle{empty}

\newpage

\thispagestyle{empty}



\setcounter{page}{1}

\section{Introduction}

Real-world networks are prone to failures. The fundamental goal of \emph{survivable network design} is to build low-cost networks that provide the desired connectivity between pairs or groups of nodes despite the failure of a few edges or nodes. One of the most basic survivable network design problems is the 2-Edge-Connected Spanning Subgraph problem (2ECSS): we are given an undirected, unweighted graph $G=(V,E)$. A feasible solution is a subset of edges $S\subseteq E$ such that the subgraph $G'=(V,S)$ is 2-edge-connected (2EC)\footnote{We recall that a graph is $k$-edge-connected (kEC) if it remains connected after the removal of an arbitrary subset of $k-1$ edges.}. In this scenario, $G'$ ensures connectivity among all pairs of nodes even in the presence of a single edge fault. Our goal is to find an (optimal) feasible solution $\OPT=\OPT(G)$ of minimum cardinality (or \emph{size}) $\opt=\opt(G)$. 

2ECSS is obviously NP-hard: a graph with $n$ nodes has a Hamiltonian cycle iff it contains a 2EC spanning subgraph with $n$ edges. In fact, 2ECSS is APX-hard \cite{CL99,F98}, which rules out the existence of a PTAS for it unless $P=NP$. A significant amount of research has been devoted to designing algorithms for 2ECSS with small approximation ratios. 
It is easy to compute a $2$-approximation for this problem.
For example, one can compute a DFS tree and augment it by picking the highest back edge for each non-root vertex. Khuller and Vishkin \cite{KV94} found the first non-trivial $3/2$-approximation algorithm. Cheriyan, Seb{\"{o}}, and Szigeti \cite{CSS01} improved the approximation ratio to $17/12$. This was further improved to $4/3$ in two independent (also in terms of techniques) works by Hunkenschr{\"o}der, Vempala, and Vetta \cite{HVV19} and Seb{\"o} and Vygen \cite{SV14}. 
In a recent work, Garg, Grandoni, and Jabal Ameli \cite{GGJ23soda} obtained a $\frac{118}{89}+\eps<1.326$ approximation, based on a rather complex case analysis. Shortly afterward, Kobayashi and Noguchi \cite{KN23} observed that one can replace a 2-edge cover used in \cite{GGJ23soda} with a triangle-free 2-edge cover, simplifying the analysis in \cite{GGJ23soda} while obtaining an improved $(1.3 + \varepsilon)$-approximation. 
Until this work, this was the best known approximation ratio. The resulting shortened analysis remains very complex. Our main result is as follows.
\begin{theorem}\label{thr:main}
There is a deterministic $5/4$-approximation algorithm for 2ECSS that runs in polynomial time.
\end{theorem}

\subsection{Related Work}

The $2$-Vertex-Connected Spanning Subgraph problem (2VCSS) is the node-connectivity version of 2ECSS. The input is the same as in 2ECSS, and the objective remains to minimize the number of edges in the chosen subgraph $G'$. However, in this case, $G'$ must be 2-vertex-connected (2VC)\footnote{We recall that a graph $G'$ is $k$-vertex-connected (kVC) if it remains connected after the removal of an arbitrary subset of $k-1$ nodes. In other words, $G'$ has no vertex cut of size at most $k-1$.}. In other words, $G'$ does not contain any \emph{cut node}. A 2-approximation for 2VCSS can be obtained in different ways. For example, one can compute an open ear decomposition of the input graph and remove the \emph{trivial} ears, i.e., those consisting of a single edge. The resulting graph is 2VC and contains at most $2n-3$ edges (while the optimum solution must contain at least $n$ edges). The first non-trivial $5/3$-approximation was obtained by Khuller and Vishkin \cite{KV94}. This was improved to $3/2$ by Garg, Vempala, and Singla \cite{GVS93} (see also an alternative $3/2$-approximation by Cheriyan and Thurimella \cite{CT00}), and further to $10/7$ by Heeger and Vygen \cite{HV17}. The current best $4/3$-approximation is due to Bosch-Calvo, Grandoni, and Jabal Ameli \cite{BGJ23}.

The k-Edge-Connected Spanning Subgraph problem (kECSS) is the natural generalization of 2ECSS to any connectivity $k\geq 2$ (see, e.g., \cite{CHNSS22, CT00,GG12, GGTW09}).

A survivable network design problem related to kECSS is the $k$-Connectivity Augmentation Problem (kCAP): given a $k$-edge-connected graph $G=(V,E)$ and a collection of extra edges $L$ (\emph{links}), compute a minimum-cardinality subset of links $S$ such that $G'=(V,E\cup S)$ is $(k+1)$-edge-connected. Multiple better-than-$2$ approximation algorithms are known for $k=1$ (hence for odd $k$ by a known reduction \cite{DKL76}), i.e., for the Tree Augmentation Problem (TAP) \cite{A19,CTZ21,CG18,CG18a,EFKN09,FGKS18,GKZ18,KN16,KN16b,N03}. The first such approximation was obtained only recently for an arbitrary $k$ by Byrka, Grandoni, and Jabal Ameli \cite{BGJ20} (later improved substantially in \cite{CTZ21}, see also \cite{GGJS19}). Grandoni, Jabal Ameli, and Traub \cite{GJT22} presented the first better-than-$2$ approximation for the Forest Augmentation Problem (FAP), i.e., the generalization of TAP where the input graph $G$ is a forest rather than a tree. Better approximation algorithms are known for the Matching Augmentation Problem (MAP), i.e., the special case of FAP when the input forest is a matching \cite{BDS22,CCDZ23,CDGKN20,GHM23}. The best known approximation ratio of $13/8$ for MAP is due to Garg, Hommelsheim, and Megow~\cite{GHM23}.
Better-than-$2$ approximation ratios are also known for the vertex-connectivity version of TAP \cite{AHS23,N20waoa}.

For all the mentioned problems, one can naturally define a weighted version. Here, a general result by Jain \cite{J01} gives a $2$-approximation, and this was the best known approximation ratio until very recently. In a recent breakthrough \cite{TZ21}, Traub and Zenklusen presented a $1.694$-approximation for the weighted version of TAP (later improved to $1.5+\eps$ by the same authors \cite{TZ22}). Partial results in this direction were achieved earlier in \cite{A19,CN13,FGKS18,GKZ18,N17}. Even more recently, Traub and Zenklusen obtained a $(1.5+\eps)$-approximation algorithm for the weighted version of kCAP \cite{TZ23}. Finding a better-than-$2$ approximation algorithm for the weighted version of 2ECSS remains a major open problem in the area.

\section{Overview of our Approach}
\label{sec:overview}

We use standard graph notation. Given a subset of edges $F\subseteq E$, we interchangeably use $F$ and the corresponding subgraph $G'=(W,F)$, where $W = \{v\in V(G) \mid v\in f \text{ for some } f \in F\}$. The meaning will be clear from the context. For example, we might say that $F$ is 2EC and denote by $|G'|$ the number $|F|$ of its edges. We sometimes denote paths and cycles as sequences of edges $e_1\ldots e_{h-1}$ or nodes $v_1\ldots v_h$.

At a high level, our approach is similar to \cite{GGJ23soda,KN23} and also vaguely similar to \cite{BGJ23, CCDZ23,CDGKN20,GHM23}, which address related problems. The first step of our construction (details in Section \ref{sec:reductionStructured}) is a reduction to a special class of instances with a specific structure. We use the notion of \emph{structured} graphs introduced in \cite{GGJ23soda}, with minor adaptations (see also Figure \ref{fig:structured}).
\begin{definition}
    Given a real number $\alpha > 1$, a graph $G$ is $\alpha$-structured if it is simple, 2VC, contains at least $\frac{4}{\alpha - 1}$ nodes, and satisfies the following conditions:
    \begin{enumerate}\itemsep0pt
        \item  It does not contain \emph{$\alpha$-contractible} subgraphs with at most $\frac{2}{\alpha - 1}$ nodes. A subgraph $C$ is $\alpha$-contractible if it is 2EC and every 2EC spanning subgraph of $G$ contains at least $|E(C)|/\alpha$ edges of $G[V(C)]$;
        \item It does not contain any \emph{irrelevant} edge. An edge $uv$ is irrelevant if $\{u,v\}$ is a $2$-vertex cut of $G$;
        \item All its $2$-vertex cuts are \emph{isolating}. A $2$-vertex cut $\{u,v\}$ is isolating if $G\setminus \{u,v\}$ has exactly two connected components, one of which contains exactly $1$ (isolated) node.
    \end{enumerate}
\end{definition}

\begin{figure}
    \centering
    \includegraphics[width=\linewidth]{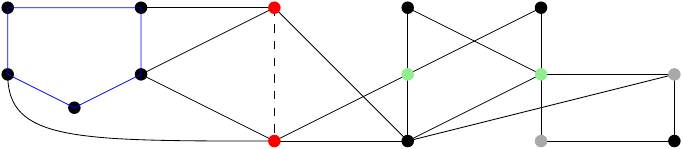}
    \caption{This figure is similar to figures in~\cite{BGJ23} and~\cite{GGJ23soda}, where the authors also use similar notions of $\alpha$-structured graphs. The subgraph induced by the blue edges is a $5/4$-contractible graph. The red and green (resp. gray) pairs of vertices form a non-isolating (resp. isolating) cut. The dashed edge is irrelevant.}
    \label{fig:structured}
\end{figure}

In \cite{GGJ23soda}, it is shown how to turn an $\alpha\geq 6/5$ approximation for 2ECSS on $\alpha$-structured graphs into an $(\alpha+\eps)$-approximation for 2ECSS in general, for any constant $\eps>0$. We use a slightly refined reduction that allows us to obtain a clean $5/4$-approximation algorithm, without an additive constant. Our reduction is also more efficient: it results in an overall running time of $n^{O(1)}$ instead of $n^{O(1/\eps)}$ as in \cite{GGJ23soda}.
\begin{lemma}\label{lem:preprocessing}
    For every constant $\alpha\geq 6/5$, if there exists a deterministic polynomial-time algorithm for 2ECSS on $\alpha$-structured graphs $G$ that returns a solution of size at most $\alpha\cdot\opt(G) - 2$, then there exists a deterministic polynomial-time $\alpha$-approximation algorithm for 2ECSS.
\end{lemma}
One key property of structured graphs (for any $\alpha$) is the following $3$-Matching Lemma, which we will also need. 
\begin{lemma}[3-Matching Lemma~\cite{GGJ23soda}]\label{lem:matchingOfSize3}
Let $G=(V,E)$ be a 2VC simple graph without irrelevant edges and without non-isolating 2-vertex cuts. Consider any partition $(V_1,V_2)$ of $V$ such that for each $i\in\{1,2\}$, $|V_i|\geq 3$ and if $|V_i|=3$, then $G[V_i]$ is a triangle. Then, there exists a matching of size $3$ between $V_1$ and $V_2$ in $G$.
\end{lemma}

We next assume that our graph $G$ is $5/4$-structured. Recall that $G$ contains at least $16$ vertices; hence $\opt(G)\geq 16$. The main property of $G$ that we will need in this overview is that it does not contain $5/4$-contractible subgraphs (along with the 3-Matching Lemma, which implicitly exploits the other properties).

The second step of our construction (details in Section \ref{sec:lowerBound}) is the computation of a proper lower bound graph. Following \cite{KN23}, and unlike \cite{GGJ23soda}, we use a minimum triangle-free 2-edge cover $H$. A $2$-edge cover $H$ of a graph $G$ is a subset of edges $H\subseteq E(G)$ such that every node of $G$ is incident to at least $2$ edges of $H$. We say that $H$ is triangle-free if no connected component of $H$ is a triangle\footnote{We remark that a connected component with at least $4$ vertices is allowed to contain a triangle as an induced subgraph.}. For ease of exposition, in the following, we refer to a connected component of a $2$-edge cover $H$ simply as a component. 
Notice also that $|H|\leq \opt$, because every feasible solution to 2ECSS is a 2-edge cover and is triangle-free, since it is connected and contains more than $3$ vertices.

For technical reasons related to the final step of our algorithm, we impose that $H$ contains at least one component with at least $8$ vertices. This can be ensured as follows. We guess a subset $F$ of edges from an optimal solution $\OPT$ that induces a connected graph spanning $8$ vertices by considering all possible (polynomially many) options. 
Next, we compute a minimum triangle-free 2-edge cover $H$ containing $F$. This constraint can be imposed easily as follows. We subdivide each edge $ab\in F$, replacing it with a path $acb$, where $c$ is a distinct dummy vertex. Let $G'$ be the resulting graph.  We then compute an (unrestricted) minimum triangle-free $2$-edge cover $H'$ of $G'$. Notice that $H'$ must contain all edges incident to dummy vertices since these vertices have degree $2$ in $G'$. Finally, we obtain a triangle-free $2$-edge cover $H$ of $G$ containing $F$ by replacing both edges $ac$ and $cb$ incident to each dummy vertex $c$ with the corresponding edge $ab$. In the final step, we exploit the fact that $F$ induces a connected graph and that $|F|\geq 4$: this guarantees that replacing the dummy edges with $F$ does not create components that are triangles. We also remark that all vertices spanned by $F$ belong to the same component of $H$, hence this component contains at least $8$ vertices. Notice that $\OPT$ induces a triangle-free $2$-edge cover of $G'$ of size $\opt+|F|$. Thus, we have $|H'|\leq \opt+|F|$, which implies $|H|\leq \opt$, as required to obtain a valid lower bound.

It is possible to impose extra useful properties on $H$ without increasing its size. A component of $H$ is \emph{complex} if it is not 2EC. A \emph{2EC block} of $H$ is a maximal 2EC subgraph of $H$.  
\begin{definition}[Canonical $2$-edge cover]
    We say that a $2$-edge cover $H$ of a graph $G$ is canonical if:
    \begin{enumerate*}[label={(\arabic*)}]
        \item $H$ is triangle-free;
        \item every component of $H$ with less than $8$ edges is a cycle;
        \item every 2EC block of a complex component of $H$ contains at least $4$ edges;
        \item every complex component contains at least two 2EC blocks with at least $6$ edges each;
        \item $H$ contains at least one component with at least $8$ vertices.
    \end{enumerate*}
\end{definition}
\begin{restatable}{lemma}{lemCanonical}\label{lem:canonical2EdgeCover}
    Let $G$ be a $5/4$-structured graph $G$, and let $H$ be a triangle-free $2$-edge cover of $G$ containing at least one component with at least $8$ vertices. Then, one can compute in polynomial time a canonical $2$-edge cover $H'$ of $G$ such that $|H'|\leq |H|$.
\end{restatable}

The third step of our construction consists of several iterations. We start with $S=H$, and we gradually transform $S$ by adding and sometimes deleting edges. At the end of the process, the final $S$ is the desired approximate solution. To control the size of $S$ during the process, we use a \emph{credit-based} argument. The idea is to assign a certain amount of credits $\credit(S)$ to $S$ based on its structure. Specifically, we define the following \emph{credit-assignment scheme}:
\begin{invariant}[Credit assignment scheme]\label{inv:creditInvariant}
    Let $S$ be a canonical $2$-edge cover of a $5/4$-structured graph $G$. Then:
    \begin{enumerate}\itemsep0pt
    \item Let $C$ be a 2EC component of $S$ with $|C|< 8$ (\emph{small  component}). Then $\credit(C) = |C|/4$.
    \item Let $C$ be a 2EC component of $S$ with $|C|\geq 8$ (\emph{large  component}). Then $\credit(C) = 2$.
    \item Let $C$ be a complex component of $S$ (i.e., not 2EC). Then $\credit(C) = 1$.
    \item Let $B$ be a 2EC block of a complex component $C$ of $S$. Then $\credit(B) = 1$.
    \item Let $b$ be a bridge\footnote{A \emph{bridge} of a graph $G'$ is an edge whose removal increases the number of connected components of $G'$.} of a complex component $C$ of $S$. Then $\credit(b) = 1/4$.
\end{enumerate}

We denote by $\credit(S)$ the total amount of the above credits assigned to the components, 2EC blocks, and bridges of $S$.

\end{invariant}
We define the \emph{cost} of $S$ as $\cost(S)=|S|+\credit(S)$. 
We can guarantee the following upper bound on the initial cost of $S$ (proof in Section \ref{sec:lowerBound}).    
\begin{restatable}{lemma}{lemCostCanonical}\label{lem:costCanonical}
    Let $H$ be a canonical $2$-edge cover of a $5/4$-structured graph $G$. Then, $\cost(H)\leq \frac{5}{4} |H|$. 
\end{restatable}

Throughout the iterations, we maintain the invariant that the cost of $S$ never increases. That is, at each step, we replace $S$ with $S'$ ensuring that $\cost(S')\leq \cost(S)$. 
The transformation process of $S$ consists of two main stages. The first stage, called \emph{bridge covering}, handles all the bridges of $S$. Here, we essentially replicate the analysis in \cite{GGJ23soda}, with minor modifications, to obtain the following (proof in Section~\ref{sec:bridgeCovering}).
\begin{restatable}[Bridge Covering]{lemma}{lemBridgeCovering}\label{lem:bridgeCovering}
Let $S$ be a canonical 2-edge cover of a $5/4$-structured graph $G$. Then, one can compute in polynomial time a canonical 2-edge cover $S'$ of $G$ such that all components of $S'$ are 2EC and $\cost(S')\leq \cost(S)$. 
\end{restatable}

The second and final stage of our construction, called \emph{gluing}, gradually merges the 2EC components of $S$ into larger 2EC components. This is the most novel part of our analysis (details in Section~\ref{sec:gluing}). Notably, the most complicated part of the analysis in \cite{GGJ23soda,KN23} is a similar gluing stage. 
Our approach significantly simplifies this part of their analysis while 
reducing the approximation ratio from $13/10$ to $5/4$.

Recall that, by the definition of a canonical 2-edge cover and the fact that $S$ is bridgeless, $S$ consists of the following two types of components:
\begin{itemize}\itemsep0pt
    \item \emph{Small components}: A small component is a $k$-cycle with $k\in \{4,5,6,7\}$. Each small component $C$ has $\credit(C)=k/4$ credits.
    \item \emph{Large components}: A large component is a 2EC component containing at least $8$ edges. A large component $C$ has $\credit(C)=2$ credits.   
\end{itemize}
The high-level idea is as follows. Consider the \emph{component (multi)-graph} $\hat{G}_S$, that is obtained by contracting each 2EC component $C$ of $S$ into a single component-node $\hat{C}$ and removing any loops (parallel edges may still exist). 
The key idea is to compute a cycle $F$ in $\hat{G}_S$ (involving nodes $\hat{C}_1,\dots,\hat{C}_{|F|}$). 
We add $F$ to $S$, replacing $S$ with $S'=S\cup F$. As a result, all components $C_i$ merge into a single (large) 2EC component $C'$ of $S'$. The decrease in the solution cost (which must be non-negative) is given by
\begin{align*}    
\cost(S)-\cost(S') &= |S|-|S'|+\sum_{j=1}^{|F|}\credit(C_i)-\credit(C') \\
&=-|F|+\sum_{j=1}^{|F|}\credit(C_i)-2 \ .
\end{align*}

In many cases, this simple approach is sufficient. For example, if all the $C_i$'s are large (each owning $2$ credits), then, for any $|F|\geq 2$, we have $\cost(S')\leq \cost(S)$, as desired. 

The most problematic case arises when $|F|$ is small and involves small components of $S$, which own fewer credits. Here, we need to delete some edges from $S$; otherwise, the credits of the $C_i$'s are insufficient. A typical situation when we can perform such a deletion occurs when two edges $e,f\in F$ are adjacent to consecutive nodes $u$ and $v$ of some small component $C_i$. In this case, we can remove the edge $uv$ from $S'$ while keeping $C'$ 2EC. 
In other words, we can set $S'\leftarrow S' \setminus \{uv\}$. 
In other cases, we can ensure that $u$ and $v$, though not adjacent in $C_i$, are connected via a Hamiltonian path $P$ 
that spans all vertices $V(C_i)$ of $C_i$. In this case, we can still save one edge by setting $S'\leftarrow (S' \setminus E(C_i)) \cup E(P)$. In some cases, we can delete multiple edges this way. A key aspect of the case analysis is proving that any increase in $|S|$ can always be compensated by a corresponding decrease in $\credit(S)$. Specifically, we will prove the following lemma. 
\begin{restatable}[Gluing]{lemma}{lemgluingSimplified}\label{lem:gluingSimplified}
    Given a disconnected canonical $2$-edge cover $S$ of a $5/4$-structured graph $G$ such that every component of $S$ is 2EC. One can compute in polynomial time  a canonical $2$-edge cover $S'$ of $G$ such that every component of $S'$ is 2EC, $S'$ has fewer components than $S$, and $\cost(S')\leq \cost(S)$.
\end{restatable}
The $5/4$-approximation algorithm then easily follows.
\begin{proof}[Proof of Theorem \ref{thr:main}]
    By Lemma~\ref{lem:preprocessing}, it suffices to describe an algorithm that, given a $5/4$-structured graph $G$, finds a solution of cost at most $5/4 \cdot\opt(G)-2$. We guess a set $F$ of edges belonging to some optimal solution $\OPT$ and induce a connected graph with $8$ nodes. We then compute a minimum triangle-free $2$-edge cover $\hat{H}$ of $G$ that contains $F$, as described earlier. We transform $\hat{H}$ into a canonical $2$-edge cover $H$ of no larger size using Lemma \ref{lem:canonical2EdgeCover}. Recall that $|H|\leq \opt(G)$.     
    
    By first applying Lemma~\ref{lem:bridgeCovering} and then iteratively applying Lemma~\ref{lem:gluingSimplified} a polynomial number of times, we obtain a 2-edge cover $\APX$ consisting of a single 2EC component, which we return as the feasible solution. By construction, $\cost(\APX)\leq \cost(H)$. By Lemma~\ref{lem:costCanonical}, $\cost(H)\leq 5/4 \cdot |H|\leq 5/4 \cdot\opt(G)$. Since $\APX$ contains exactly one large 2EC component, we have $|\APX|=\cost(\APX) - \credit(\APX) = \cost(\APX) - 2$, and the claim follows.
\end{proof}

It remains to prove Lemma \ref{lem:gluingSimplified}. The key insight of our approach is that it is beneficial to focus on the 2VC blocks\footnote{A 2VC block $B$ of a graph $G'$ is a maximal subgraph that is 2VC.} $B$ of $\hat G_S$. More precisely, we will consider one such block that contains a large component of $S$. Recall that, by the definition of canonical, $S$ must contain a component with at least $8$ vertices and, since $S$ is a $2$-edge cover,  with at least $8$ edges. Thus, a large component is guaranteed to exist. A crucial tool in our analysis is the following \emph{local} version of the 3-Matching Lemma, which applies to a single 2VC block of $\hat G_S$.
\begin{lemma}[Local $3$-Matching Lemma]\label{lem:local3matching}
Let $B$ be a 2VC block of $\hat{G}_S$ and let $(\hat V_1, \hat V_2)$ be a partition of the nodes of $B$ such that $\hat V_1 \neq \emptyset \neq \hat V_2$. Let $V_i=\cup_{\hat{C}_i\in \hat{V}_i}V(C_i)$ be the set of vertices of $G$ that correspond to $\hat V_i$, for $i\in\{1, 2\}$. Then, there is a matching of size $3$ between $V_1$ and $V_2$ in $G$.
\end{lemma}
\begin{proof}
Let $\hat V_i'$ be the set of nodes in $V(\hat G_S)\setminus V(B)$ connected to $\hat V_i$ by some path contained in $\hat G_S\setminus E(B)$, for $i\in\{1, 2\}$. Since $B$ is a maximal 2VC subgraph of $\hat G_S$, $\hat V_1'$ and $\hat V_2'$ are disjoint. Moreover, nodes in $\hat V_1'$ are not adjacent to nodes of $\hat V_2'$, as that would imply they also belong to $\hat V_2'$. Thus, there are no edges from $\hat V_1'$ to $\hat V_2\cup\hat V_2'$ or from $\hat V_2'$ to $\hat V_1\cup\hat V_1'$. Let $V_1'$ and $ V_2'$ be the set of nodes of $G$ corresponding to $\hat V_1'$ and $\hat V_2'$, resp.

The 3-Matching Lemma (Lemma~\ref{lem:matchingOfSize3}) guarantees that there is a matching $M$ of size $3$ between $V_1\cup V'_1$ and $V_2\cup V'_2$. However, the vertices in $V'_1$ (resp., $V'_2$) are not adjacent to those in $V_2\cup V'_2$ (resp., $V_1\cup V'_1$). As a result, the endpoints of $M$ must be in $V_1 \cup V_2$. 
\end{proof}

\section{Gluing}
\label{sec:gluing}

This section is devoted to 
proving 
Lemma~\ref{lem:gluingSimplified}. Before delving into the proof, 
we first establish a series of structural lemmas that hold for structured graphs. In particular, many of these lemmas rely on the structure of the 2VC blocks of $\hat G_S$.  
Specifically, we begin with the following lemmas that show that small components are ``well connected'' to the rest of the graph.  

\begin{observation}\label{obs:relationContractability}
For $\alpha\geq \beta\geq 1$, a 2EC subgraph $C$ of a 2EC graph $G$ that is $\beta$-contractible is also $\alpha$-contractible.
\end{observation}
\begin{lemma}\label{lem:6-cycleContractibility}
    Let $G$ be a $5/4$-structured graph, and let $C=x_1x_2x_3x_4x_5x_6$ be a $6$-cycle in $G$.  
    If there is at most one pair of distinct vertices $\{x,x'\}\subseteq\{x_1, x_3, x_5\}$ that has a $x,x'$-Hamiltonian path in $G[V(C)]$, then there are at least two distinct edges in $G\setminus G[V(C)]$ incident to nodes in $\{u_2,u_4,u_6\}$.
\end{lemma}
\begin{proof}
    Assume there is at most one such edge in $G$. We first show that there is no edge between any pair of nodes in $\{x_2, x_4, x_6\}$. W.l.o.g.\ assume that $x_2x_6\in E(G)$. Then the paths $x_1x_2x_6x_5x_4x_3$ and $x_1x_6x_2x_3x_4x_5$ are Hamiltonian paths in $G[V(C)]$ between $x_1, x_3$ and $x_1, x_5$, respectively, a contradiction. 
    
    Thus, every 2EC spanning subgraph of $G$ contains at least $5$ edges of $G[V(C)]$, as it must contain at least $2$ edges incident to each vertex in $\{x_2,x_4,x_6\}$. Therefore, $C$ is $\frac{6}{5}$-contractible, a contradiction to the fact that $G$ is $\frac{5}{4}$-structured.
\end{proof}

\begin{lemma}[Hamiltonian Pairs Lemma]\label{lem:hamiltonianPairs}
    Let $G$ be a $5/4$-structured graph, and let $C$ be a $k$-cycle in $G$. A Hamiltonian pair $\{u,v\}$ of $C$ consists of distinct nodes $u,v\in V(C)$ such that: (1) they are both adjacent to nodes in $V(G)\setminus V(C)$ and (2) there is a Hamiltonian $u, v$-path in $G[V(C)]$. If $k \leq 6$, then $C$ has at least $2$ distinct Hamiltonian pairs $\{u_1,v_1\},\{u_2,v_2\}$, in particular $|\{u_1,v_1\}\cap \{u_2,v_2\}|\leq 1$. If $k=7$, then $C$ has at least one such pair.      
\end{lemma}

\begin{proof}
    Let $M = \{u_1v_1, u_2v_2, u_3v_3\}$ be  a matching 
    obtained by applying the $3$-Matching Lemma~\ref{lem:matchingOfSize3} to the partition $(V(C), V(G)\setminus V(C))$, 
    where $u_i\in V(C)$ and $v_i\in V(G)\setminus V(C)$ for all $i\in\{1, 2, 3\}$ (observe that the partition satisfies the conditions of Lemma~\ref{lem:matchingOfSize3}). Notice that, for $k=4$, at least one node $u_i$, say $u_3$, is adjacent to the other two nodes $u_1$ and $u_2$. Hence the pairs $\{u_1,u_3\}$ and $\{u_2,u_3\}$ satisfy the claim. 
    
    We next assume that $C=x_1x_2\dots x_k$ with $k\in\{5, 6, 7\}$. We can assume w.l.o.g.\ that there is at most one Hamiltonian pair  
    with both its vertices from $\{u_1, u_2, u_3\}$, otherwise we are done. In particular, $\{u_1,u_2,u_3\} \neq \{x_1,x_2,x_3\}$. Moreover, if such a pair exists, we assume w.l.o.g.\ that it is $\{u_1,u_2\}$ (while $\{u_1,u_3\}$ and $\{u_2,u_3\}$ are not).
    We now consider different cases based on the value of $k$:
    
    \begin{caseanalysis}
    \case{$k=5$.} We can assume w.l.o.g.\ that $u_1 = x_1, u_2 = x_2, u_3 = x_4$. $\{x_1,x_2\}$ is one of the desired Hamiltonian pairs. We can assume that $x_3$ (resp., $x_5$) has no neighbors in $V(G)\setminus V(C)$, otherwise the second pair is $\{x_2,x_3\}$ (resp., $\{x_1,x_5\}$). 
    Now, the edge $x_3x_5$ must exist in $G$, otherwise every 2EC subgraph of $G$ must 
    contain at least $4$ edges of $G[V(C)]$, $2$ for $x_3$, and $2$ for  $x_5$,   
    making $C$ 
    $5/4$-contractible, a contradiction to the fact that $G$ is $5/4$-structured. Then, the path $x_1x_2x_3x_5x_4$ is a Hamiltonian $x_1, x_4$-path in $G[V(C)]$, and hence $\{x_1,x_4\}$ is the other desired Hamiltonian pair.

    \case{$k=6$.}
    
    \subcase{$u_1$ is adjacent to $u_2$ in $C$.} We can assume w.l.o.g.\ that $u_1 = x_1, u_2 = x_2, u_3 = x_4$. 
    $\{x_1, x_2\}$ is one of the desired pairs.
    Moreover, we can assume that $x_3$, $x_5$, and $x_6$ have no neighbors in $V(G)\setminus V(C)$, otherwise $\{x_3,x_2\}$, $\{x_5,x_4\}$, and $\{x_6,x_1\}$, respectively, is the other desired pair. Notice that $x_3$ must be adjacent to either $x_5$ or $x_6$ in $G$, otherwise every 2EC subgraph of $G$ must 
    contain at least $5$ edges of $G[V(C)]$ ($2$ for $x_3$ and $3$ in total for $x_5$ and $x_6$), making $C$  
    $6/5$-contractible, and hence $5/4$-contractible by Observation~\ref{obs:relationContractability}, a contradiction to the fact that $G$ is $5/4$-structured. If 
    $x_3x_5\in E(G)$, then $x_2x_1x_6x_5x_3x_4$ is a Hamiltonian $x_2, x_4$-path in $G[V(C)]$, and hence $\{x_2,x_4\}$ is the other desired Hamiltonian pair. Similarly, if  
    $x_3x_6\in E(G)$, then $\{x_1,x_4\}$ is the other desired pair.
    
    \subcase{$u_1$ is not adjacent to $u_2$ in $C$.}
    We can assume w.l.o.g.\ $u_1 = x_1, u_2 = x_3, u_3 = x_5$. Moreover, we can assume that $x_2$, $x_4$, and $x_6$ have no neighbors in $V(G)\setminus V(C)$, otherwise the two desired pairs are $\{\{x_2,x_1\},\{x_2,x_3\}\}$, $\{\{x_4,x_3\},\{x_4,x_5\}\}$, and $\{\{x_6,x_5\},\{x_6,x_1\}\}$, respectively. Now, applying Lemma~\ref{lem:6-cycleContractibility}, either there must be two Hamiltonian pairs with vertices from $\{x_1,x_3,x_5\}$ and we are done or there must exist edges $xy$ where $x\in\{x_2,x_4,x_6\}$ and $y\notin V(C)$, a contradiction.

    \case{$k=7$.} In this case, we only need to prove the existence of one Hamiltonian pair. 
    Hence, we can assume that $u_1$, $u_2$, and $u_3$ are not pairwise adjacent in $C$. Thus, w.l.o.g., we can assume $u_1 = x_1$, $u_2 = x_3$, and $u_3 = x_5$. Moreover, we can assume that $x_2$, $x_4$, $x_6$, and $x_7$ do not have neighbors in $V(G)\setminus V(C)$, otherwise the pair $\{x_2,x_1\}$, $\{x_4,x_5\}$, $\{x_6,x_5\}$, and $\{x_7,x_1\}$, resp., satisfies the claim. If $E(G[\{x_2, x_4, x_6, x_7\}]) = \{x_6x_7\}$, then every 2EC subgraph of $G$ must use at least $7$ edges of $G[V(C)]$ ($2$ for node $x_2$, $2$ for node $x_4$, and $3$ in total for the nodes $x_6$ and $x_7$ together).  Then $C$ is $1$-contractible, hence $5/4$-contractible by Observation \ref{obs:relationContractability}, contradicting the fact that $G$ is $5/4$-structured. Thus there is at least another edge besides $x_6x_7$ between two nodes in $\{x_2, x_4, x_6, x_7\}$. If $x_2x_4\in E(G)$, then $x_3x_4x_2x_1x_7x_6x_5$ is a Hamiltonian $x_3, x_5$-path in $G[V(C)]$, and thus, $\{x_3,x_5\}$ is the desired pair. If $x_2x_6\in E(G)$, then $x_1x_7x_6x_2x_3x_4x_5$ is a Hamiltonian $x_1, x_5$-path in $G[V(C)]$, and thus, $\{x_1,x_5\}$ is the desired pair. If $x_2x_7\in E(G)$, then $x_1x_2x_7x_6x_5x_4x_3$ is a Hamiltonian $x_1, x_3$-path in $G[V(C)]$, and thus, $\{x_1,x_3\}$ is the desired pair. The remaining cases follow symmetrically.
    \end{caseanalysis}
\end{proof}

In the following, we assume we are given a canonical $2$-edge cover $S$ of $G$, and we focus on the 2VC blocks of the component graph $\hat G_S$. First 
we 
introduce some terminology. We say that a component $C$ of $S$ is \emph{local} if $\hat C$ belongs to exactly one 2VC block of $\hat G_S$; otherwise, it is \emph{non-local}. That is, a component $C$ of $S$ is non-local if and only if $\hat C$ is a cut vertex of $\hat G_S$.

Given two components $C_1, C_2$ of $S$ and a node $u_1\in V(C_1)$, we use the notation $u_1\hat C_2$ to indicate an edge between $u_1$ and some $u_2\in V(C_2)$ (specifically, when identifying such a $u_2$ is not required in our arguments). In this case, we also say that $u_1$ is adjacent to $\hat C_2$. Similarly, we use the notation $\hat C_1\hat C_2$ to indicate an edge between some $u_1\in V(C_1)$ and some $u_2\in V(C_2)$.

Given a 2VC block $B$ of $\hat G_S$, and $\hat{C}\in V(B)$, we say that $C$ is a component of $B$. Furthermore, we say that we apply the local $3$-Matching Lemma~\ref{lem:local3matching} to $C$ to mean that we apply it to $B$ and to the partition $(\{\hat C\}, V(B)\setminus \{\hat C\})$. 

We now present some lemmas that allow us to find cycles in $\hat G_S$ with convenient properties.

\begin{lemma}\label{lem:niceCycle}
    Let $B$ be a 2VC block of $\hat{G}_S$ and let $C_1$ and $C_2$ be two 
    distinct components of $B$. Given edges $u_1\hat X_1, u_2\hat X_2$, where $u_1\in V(C_1), u_2\in V(C_2), \hat X_1\in V(B)\setminus\{\hat C_1\}, \hat X_2\in V(B)\setminus\{\hat C_2\}$, one can compute in polynomial time  a cycle $F$ in $B$ containing $\hat C_1$ and $\hat C_2$, such that $F$ is incident on $u_1$ and another node in $C_1$, and incident on $u_2$ and another node in $C_2$.
\end{lemma}

\begin{proof}
    We start by 
    computing a cycle $F$ in $B$ that contains the component-nodes $\hat C_1, \hat C_2$, which can be done efficiently since $B$ is 2VC and both $\hat C_1$ and $\hat C_2$ belong to $B$. Let $F$ be the union of $2$ internally vertex-disjoint paths $P_v, P_w$ between $\hat C_1$ and $\hat C_2$, where $P_v$ is incident to $v_i\in V(C_i)$ and $P_w$ is incident to $w_i \in V(C_i)$, for $i\in\{1, 2\}$.
    We remark that we might have $v_i=w_i$ for some $i\in\{1, 2\}$. 
    Let  $P_v = \hat C^v_1\hat C^v_2\dots\hat C^v_{k_v}$ and $P_w = \hat C^w_1\hat C^w_2\dots\hat C^w_{k_{w}}$, with $\hat C^v_1 = \hat C^w_1 = \hat C_1$ and $\hat C^v_{k_v} =  \hat C^w_{k_w} = \hat C_2$. 
    
    We first show that we can assume that $v_1=u_1$. Assume this is not the case, i.e., $u_1\notin\{v_1,w_1\}$. 
    Since $B$ is 2VC, we can find in polynomial time a path $P_{X_1}$ in $\hat G_S$ from $\hat X_1$ to a node in $V(F)\subseteq V(B)$ not going through $\hat C_1$. Assume w.l.o.g.\ that $P_{X_1}$ has $\hat C^v_i, 1 < i \leq k_v$, as an endpoint. Then the path $\{u_1\hat X_1\}\cup P_{X_1}\cup\hat C^v_i \hat C^v_{i+1}\dots\hat C^v_{k_v}$ is a path between $\hat C_1$ and $\hat C_2$, internally vertex-disjoint with $P_w$ and incident to $u_1$ in $C_1$.

    We now show that we can assume $v_1 = u_1$ and $w_1 \neq u_1$. To get a contradiction, assume $v_1=u_1=w_1$. 
    By the local $3$-Matching Lemma (Lemma~\ref{lem:local3matching}), there must exist an edge $u\hat X$, 
    where $u\in V(C_1)\setminus\{u_1\}$ and $ \hat X\in V(B)\setminus\{\hat C_1\}$. Since $B$ is 2VC, we can find in polynomial time a path $P_X$ in $\hat G_S$ from $\hat X$ to a node in $V(F)\subseteq V(B)$ not going through $\hat C_1$. Assume w.l.o.g.\ that $P_X$ has $\hat C^w_i, 1 < i \leq k_w$, as an endpoint. Then the path $\{u\hat X\}\cup P_X\cup\hat C^w_i \hat C^w_{i+1}\dots\hat C^w_{k_w}$ is a path between $\hat C_1$ and $\hat C_2$, incident to $u\neq u_1$ in $C_1$, and internally vertex-disjoint with $P_v$ which is incident to $u_1$ in $C_1$.  

    We have shown that $v_1 = u_1$ and $w_1\neq u_1$.
    Furthermore, we now show that we can assume 
    either $v_2=u_2$ or $w_2=u_2$. Assume this is not the case, i.e., $u_2\notin\{v_2,w_2\}$.  
    Then, since $B$ is 2VC, we can find in polynomial time a path $P_{X_2}$ in $\hat G_S$ from $\hat X_2$ to a node in $V(F)\subseteq V(B)$ not going through $\hat C_2$. Assume first that $P_{X_2}$ has $\hat C^w_i, 1 \leq i < k_w$, as an endpoint, and if $i=1$, then $P_{X_2}$ is not incident to $u_1$. Then the path $\{u_2\hat X_2\}\cup P_{X_2}\cup\hat C^w_i \hat C^w_{i-1}\dots\hat C^w_1$ is a path between $\hat C_2$ and $\hat C_1$, incident to $u_2$ in $C_2$, not incident to $u_1$ in $C_1$, and internally vertex-disjoint with $P_v$. Since $P_v$ is incident to $u_1$ in $C_1$ the claim follows. Assume now $P_{X_2}$ has $\hat C^v_i, 1 \leq i < k_v$, as an endpoint, and if $i=1$, then $P_{X_2}$ is incident to $u_1$. Then the path $\{u_2\hat X_2\}\cup P_{X_2}\cup\hat C^v_i \hat C^v_{i-1}\dots\hat C^v_1$ is a path between $\hat C_2$ and $\hat C_1$, incident to $u_1$ in $C_1$, incident to $u_2$ in $C_2$, and internally vertex-disjoint with $P_w$. Since $P_w$ is not incident to $u_1$ in $C_1$, the claim follows.

    Finally, we can now assume that $v_1 = u_1, w_1\neq u_1$, and $u_2=v_2=w_2$, otherwise we are done. By the local $3$-Matching Lemma~\ref{lem:local3matching}, there must exist an edge $u\hat X$ with $u\in V(C_2)\setminus\{u_2\}$ and $\hat X\in V(B)\setminus\{\hat C_2\}$. Since $B$ is 2VC, we can find in polynomial time a path $P_X$ in $\hat G_S$ from $\hat X$ to a node in $V(F)\subseteq V(B)$ not going through $\hat C_2$. Assume first that $P_X$ has $\hat C^w_i, 1 \leq i < k_w$, as an endpoint, and if $i=1$, then $P_X$ is not incident to $u_1$. Then the path $\{u\hat X\}\cup P_X\cup\hat C^w_i\hat C^w_{i-1}\dots\hat C^w_1$ is a path between $\hat C_2$ and $\hat C_1$, not incident to $u_1$ in $C_1$, not incident to $u_2$ in $C_2$, and internally vertex-disjoint with $P_v$. Since $P_v$ is incident to $u_1$ in $C_1$ and $u_2$ in $C_2$, the lemma follows. Assume now that $P_X$ has $\hat C^v_i, 1 \leq i < k_v$, as an endpoint, and if $i=1$, then $P_X$ is incident to $u_1$. Then the path $\{u\hat X\}\cup P_X\cup\hat C^v_i\hat C^v_{i-1}\dots\hat C^v_1$ is a path between $\hat C_2$ and $\hat C_1$, incident to $u_1$ in $C_1$, not incident to $u_2$ in $C_2$, and internally vertex-disjoint with $P_w$. Since $P_w$ is incident to $w_1\neq u_1$ in $C_1$ and $u_2$ in $C_2$, the lemma follows. 
\end{proof}

\begin{corollary}\label{cor:cycleSize3}
    Let $B$ be a 2VC block of $\hat{G}_S$ and let $C_1, C_2$ be two different components of $B$. Given an edge $u_1\hat X$, $u_1\in V(C_1), \hat X\in V(B)\setminus\{\hat C_1\}$, one can compute in polynomial time a cycle $F$ in $B$ containing $\hat C_1$ and $\hat C_2$ such that $F$ is incident on $u_1$ and another node in $C_1$ 
    and $|F|\geq \min\{3, |V(B)|\}$.
\end{corollary}
\begin{proof}
    Let $F'$ be the cycle in $B$ found by applying Lemma~\ref{lem:niceCycle} to $C_1$ and the edge $u_1\hat X$ and $C_2$ and an arbitrary edge in $B$ incident on $C_2$. Let $u_1, v_1$ be the distinct nodes of $C_1$ incident with $F'$. We can assume that $|F'|<\min\{3, |V(B)|\}$, otherwise we are done. Thus, $|F'|=2$ and $|V(B)|\geq 3$, and $F' = \{u_1u_2, v_1v_2\}$, where $u_2, v_2\in V(C_2)$. Since $B$ is 2VC and $|V(B)|\geq 3$, it must contain an edge $w_1\hat C$, where $w_1\in V(C_1)$ and $\hat C\in V(B)\setminus\{\hat C_1, \hat C_2\}$. Because $B$ is 2VC, there must exist a path $P$ in $B$ from $C_{2}$ to $C$ not going through $C_{1}$. If $w_1\neq u_1$, let $F:=P\cup \{u_1u_2, w_1\hat C\}$. This way, $F$ is a cycle in $B$ incident to distinct nodes $u_1, w_1$ of $C_1$, and $|F|\geq 3$. Otherwise, if $w_1=u_1$, $F:=P\cup \{v_1v_2, w_1\hat C\}$ is the desired cycle.
\end{proof}

\begin{lemma}\label{lem:shortcutC4localC5}
    Let $B$ be a 2VC block of $\hat{G}_S$ and let $C_1, C_2$ be two different components of $B$ such that $C_1$ is a $4$-cycle or a local $5$-cycle. Given an edge $u_2\hat X$, where $u_2\in V(C_2)$ and $ \hat X\in V(B)\setminus\{\hat C_2\}$, one can compute in polynomial time  a cycle $F$ in $B$ incident to distinct nodes $u_i, v_i\in V(C_i)$ for $i\in\{1, 2\}$ such that there is a Hamiltonian $u_1, v_1$-path in $G[V(C_1)]$.
\end{lemma}

\begin{proof}
We start by computing a cycle $F$ incident to distinct nodes of $C_1$ and to $u_2, v_2$ in $C_2$ via Lemma~\ref{lem:niceCycle}, where $v_2$ is a node in $V(C_2)\setminus\{u_2\}$. Let $C_1 = x_1x_2\dots x_k, k\in\{4, 5\}$. We can assume that $F$ is incident to $x_1$ and $x_3$, and there is no Hamiltonian $x_1, x_3$-path in $G[V(C)]$, otherwise we are done. Let $P_u$ and $P_v$ be the two internally vertex-disjoint paths from $C_1$ to $C_2$ such that their union is $F$, they are incident to $x_1$ and $x_3$ in $C_1$, and to $u_2$ and $v_2$ in $C_2$. Let $P_u = \hat C^u_1\hat C^u_2\dots\hat C^u_{k_{u}}$ and $P_v = \hat C^v_1\hat C^v_2\dots\hat C^v_{k_v}$, where $\hat C^u_1 = \hat C^v_1 = \hat C_1$ and $\hat C^u_{k_u} =  \hat C^v_{k_v} = \hat C_2$.

First, assume that there is an edge $y\hat Y$, where $y\in V(C_1)\setminus\{x_1, x_3\}$ and $\hat Y\in V(B)\setminus\{\hat C_1\}$ such that there is a Hamiltonian $y, x_i$-path in $G[V(C_1)]$ for each $i\in\{1,3\}$. 
Since $B$ is 2VC, we can find a path $P_Y$ in $\hat G_S$ from $\hat Y$ to a node in $V(F)\subseteq V(B)$ not going through $\hat C_1$. Assume first that $\hat C^u_i, 1 < i \leq k_u$, is an endpoint of $P_Y$, and if $i = k_u$, then $P_Y$ is incident to $u_2$. Then, the path $\{y\hat Y\}\cup P_Y\cup \hat C^u_i\hat C^u_{i+1}\dots\hat C^u_{k_{u}}$ is incident to $y$ in $C_1$, incident to $u_2$ in $C_2$, and internally vertex-disjoint with $P_v$. Since $P_v$ is incident to $x_3$ in $C_1$ and to $v_2$ in $C_2$, the lemma follows. We can now assume $\hat C^v_i, 1 < i \leq k_v$, is an endpoint of $P_Y$, and if $i = k_v$, then $P_Y$ is not incident to $u_2$. Then, the path $\{y\hat Y\}\cup P_Y\cup \hat C^v_i\hat C^v_{i+1}\dots\hat C^v_{k_{v}}$ is incident to $y$ in $C_1$, not incident to $u_2$ in $C_2$, and internally vertex-disjoint with $P_u$. Since $P_u$ is incident to $x_1$ in $C_1$ and to $u_2$ in $C_2$, the lemma follows.

We can now assume that there is no such edge. In particular, this implies that there is no edge $x_2\hat X_2$, where $\hat X_2\in V(B)\setminus\{\hat C_1\}$. Then, by the local $3$-Matching Lemma~\ref{lem:local3matching}, there must be an edge $y\hat Y$ such that $y\in V(C_1)\setminus\{x_1, x_2, x_3\}$. If $C_1$ is a $4$-cycle $y=x_4$, then, since $x_1x_4, x_3x_4\in E(C_1)$, $x_4\hat Y$ is an edge of the excluded type. Thus, it must be that $C_1$ is a $5$-cycle, and by symmetry, we can assume w.l.o.g.\ $y=x_4$. Thus, there is an edge $x_4\hat X_4$, where $\hat X_4\in V(B)\setminus\{\hat C_1\}$. By our previous assumption, there is no Hamiltonian $x_1, x_4$-path in $G[V(C_1)]$. Since there is no edge $x_2\hat X_2$, where $\hat X_2\in V(B)\setminus\{\hat C_1\}$, by the Hamiltonian pairs lemma (Lemma~\ref{lem:hamiltonianPairs}), and since $C_1$ is local by assumption of the lemma, there must also exist an edge $x_5\hat X_5$, 
where $\hat X_5\in V(B)\setminus\{\hat C_1\}$. 
    
Since $B$ is 2VC, we can find in polynomial time a path $P_{X_4}$ in $\hat G_S$ from $\hat X_4$ to a node in $V(F)\subseteq V(B)$ not going through $\hat C_1$. Assume first $P_{X_4}$ has $\hat C^u_i, 1 < i \leq k_u$, as an endpoint, and if $i = k_u$, then it is incident to $u_2$. Then $\{x_4\hat X_4\}\cup P_{X_4}\cup \hat C^u_i\hat C^u_{i+1}\dots\hat C^u_{k_u}$ is incident to $x_4$ in $C_1$, incident to $u_2$ in $C_2$, and internally vertex-disjoint with $P_v$. Since $P_v$ is incident to $x_3$ in $C_1$ and to $v_2$ in $C_2$, the lemma follows. We can then assume that $P_{X_4}$ has $\hat C^v_i$ as an endpoint for some $1 < i \leq k_v$, and if $i = k_v$, then $P_{X_4}$ is not incident to $u_2$. Similarly, we define $P_{X_5}$, and by a similar argument it has $\hat C^u_j$ as an endpoint for some $1 < j \leq k_u$, and if $j = k_u$, then $P_{X_5}$ is incident to $u_2$.

If the paths $P_{X_4}, P_{X_5}$ are not internally vertex-disjoint, then we can efficiently compute a path $P'_{X_4}$ in $\hat G_S$ from $\hat X_4$ to a node in $V(F)\subseteq V(B)$ having $C^u_j, 1< j\leq k_u$, as an endpoint and satisfying that if $j=k_u$, then $P'_{X_4}$ is incident to $u_2$, so by the same arguments as above the claim follows. Otherwise, that is, $P_{X_4}$ and $P_{X_5}$ are internally vertex-disjoint, the paths $\{x_4\hat X_4\}\cup P_{X_4}\cup \hat C^v_i\hat C^v_{i+1}\dots\hat C^v_{k_v}$ and $\{x_5\hat X_5\}\cup P_{X_5}\cup \hat C^u_j\hat C^u_{j+1}\dots\hat C^u_{k_u}$ are internally vertex-disjoint and they have endpoints $x_4$ and $x_5$ in $C_1$ and $w_2$ and $u_2$ in $C_2$, where $w_2\in V(C_2)\setminus\{u_2\}$. Thus the lemma follows also in this case. 
\end{proof}

Finally, the following lemma is crucial to deal with non-local components of $S$.

\begin{lemma}\label{lem:shortcutNonLocal}
    Let $F$ be a 2EC graph and let $C$ be a cycle of $F$. Assume that there are paths $P_{u_1v_1}$ and $P_{u_2v_2}$ in $F\setminus C$
    between nodes $u_i$ and $v_i\in V(C)$ for $i\in \{1, 2\}$. Assume $u_1v_i\in C$ for some $i\in\{1, 2\}$. Let $d_{C'}(u, v)$ denote the length of the $u, v$-path in $C':= C\setminus\{u_1v_i\}$, for every $u, v\in V(C)$. If $d_{C'}(u_1, u_i)\leq d_{C'}(u_1, v_1)$, then $F\setminus\{u_1v_i\}$ is 2EC.
\end{lemma}
\begin{proof}
    Let $F' := F\setminus\{u_1v_i\}$. Assume to get a contradiction $F'$ contains a bridge $e$. If $e\notin C$, then let $C_1, C_2$ be the two connected components of $F'\setminus \{e\}$ such that $C\setminus\{u_1v_i\}\subseteq C_1$. Since the edge $u_1v_i$ has both endpoints in $V(C_1)$, $e$ is also a bridge in $F$, a contradiction.

    Thus, we can assume that $e\in C$. Since $C'$ is a $u_1, v_i$-path, there are nodes $u, v\in V(C)$ such that $e=uv$ and $d_{C'}(v_i, u) = d_{C'}(v_i, v) + 1$. Let $P_j$ be the $u_j, v_j$-path in $C'$ for $j\in \{1, 2\}$. If $e\in P_j$ for some $j\in\{1, 2\}$, then $e$ belongs to the closed trail $P_j\cup P_{u_jv_j}$ in $F'$, a contradiction to the fact that $e$ is a bridge of $F'$. Thus it must be that $e\notin P_1\cup P_2$. Since $e\notin P_1$ and $C'$ is a $u_1, v_i$-path, one has that $d_{C'}(u_1, v_1)\leq d_{C'}(u_1, u)$. Since $e\notin P_i$ and $C'$ is a $u_1, v_i$-path, one has that $d_{C'}(u_i, v_i)\leq d_{C'}(v_i, v)$. It holds that:
    \begin{align*}
        &d_{C'}(u_1, u_i)\leq d_{C'}(u_1, v_1)\leq d_{C'}(u_1, u) = d_{C'}(u_1, v_i) - d_{C'}(v_i, u) <\\
        &d_{C'}(u_1, v_i) - d_{C'}(v_i, v) \leq d_{C'}(u_1, v_i) - d_{C'}(u_i, v_i) = d_{C'}(u_1, u_i),
    \end{align*}
    a contradiction. The first inequality above follows by the assumption of the lemma, the second inequality follows from having $d_{C'}(u_1, v_1)\leq d_{C'}(u_1, u)$, and the last inequality follows from having $d_{C'}(u_i, v_i)\leq d_{C'}(v_i, v)$. The two equalities follow from the fact that $C'$ is a $u_1, v_i$-path, and the strict inequality follows from having $d_{C'}(v_i, u) = d_{C'}(v_i, v) + 1$.
\end{proof}

\begin{corollary}\label{cor:shortcutNonLocalC5}
    Let $F$ be a 2EC graph and let $C$ be a $5$-cycle in $F$. Assume that there are paths $P_{u_1v_1}$ and $P_{u_2v_2}$ in $F\setminus C$ between nodes $u_i, v_i\in V(C)$ such that $u_i\neq v_i$ for $i\in \{1, 2\}$. If $u_2\notin \{u_1, v_1\}$, then we can find in polynomial time an edge $e\in C$ such that $F\setminus\{e\}$ is 2EC.
\end{corollary}
\begin{proof}
    If $u_j$ is adjacent to $v_j$ in $C$ for some $j\in\{1, 2\}$, then $F\setminus\{u_jv_j\}$ is 2EC by Lemma~\ref{lem:shortcutNonLocal} (with $u_1:=u_j, v_1:=v_j, u_2:=u_{3-j}, v_2:=v_{3-j}$). Otherwise, since $C$ is a $5$-cycle, $u_1$ is not adjacent to $v_1$, and $u_2\notin\{u_1, v_1\}$, therefore $u_2$ is adjacent to either $u_1$ or $v_1$, say w.l.o.g.\ to $v_1$. Also because $C$ is a $5$-cycle, $u_1$ is not adjacent to $v_1$ and $u_2\notin\{u_1, v_1\}$, it holds that the $u_1, u_2$-path in $C\setminus\{u_2v_1\}$ has length at most $2$. Thus, $F\setminus\{u_2v_1\}$ is 2EC follows from Lemma~\ref{lem:shortcutNonLocal} (with $u_1:=u_2, v_1:=v_2, u_2:=u_1$, $v_2:=v_1$) by observing that $d_{C\setminus\{u_2v_1\}}(u_2, u_1)\leq 2\leq d_{C\setminus\{u_2v_1\}}(u_2, v_2)$, where the second inequality follows from the fact that $u_2$ is not adjacent to $v_2$ in $C$.
\end{proof}

\begin{corollary}\label{cor:shortcutNonLocalC6C7}
    Let $F$ be a 2EC graph and let $C$ be a $6$-cycle or a $7$-cycle in $F$. Assume that there are paths $P_{u_1v_1}$ and $ P_{u_2v_2}$ in $F\setminus C$ between nodes $u_i, v_i\in V(C)$ such that $u_i\neq v_i$ for $i\in \{1, 2\}$. If $u_2\notin \{u_1, v_1\}$ and $u_2$ is in a shortest $u_1, v_1$-path in $C$, then we can find in polynomial time an edge $e\in C$ such that $F\setminus\{e\}$ is 2EC.
\end{corollary}
\begin{proof}
    If $u_j$ is adjacent to $v_j$ in $C$ for some $j\in\{1, 2\}$, then $F\setminus\{u_jv_j\}$ is 2EC by Lemma~\ref{lem:shortcutNonLocal} (with $u_1:=u_j, v_1:=v_j, u_2:=u_{3-j}, v_2:=v_{3-j}$). Otherwise, since $C$ is a $6$-cycle or a $7$-cycle, $u_2\notin\{u_1, v_1\}$ and $u_2$ is in a shortest $u_1, v_1$-path in $C$, it must be that $u_2$ is adjacent in $C$ to either $u_1$ or $v_1$. Say w.l.o.g.\ $u_2$ is adjacent to $v_1$ in $C$. Since $C$ is a $6$-cycle or a $7$-cycle and $u_2$ is in a shortest $u_1, v_1$-path in $C$, it holds that the $u_1, u_2$-path in $C\setminus\{u_2v_1\}$ has length at most $2$. Then, $F\setminus\{u_2v_1\}$ is 2EC follows from Lemma~\ref{lem:shortcutNonLocal} (with $u_1:=u_2, v_1:=v_2, u_2:=u_1, v_2:=v_1$) by observing that $d_{C\setminus\{u_2v_1\}}(u_2, u_1)\leq 2\leq d_{C\setminus\{u_2v_1\}}(u_2, v_2)$ 
    where, the second inequality follows from the fact that $u_2$ is not adjacent to $v_2$ in $C$.
\end{proof}

\subsection{Proof of Lemma~\ref{lem:gluingSimplified}}

We are now ready to prove Lemma~\ref{lem:gluingSimplified}, which we restate next. 

\lemgluingSimplified*

By the definition of canonical, $S$ contains a large component $C$. Let $B$ be any 2VC block of $\hat G_S$ that contains $C$ (if there are multiple such blocks, we choose $B$ arbitrarily). We separate the proof into a few lemmas.

\begin{lemma}\label{lem:gluingAdjacent}
    Let $C_1$ and $C_2$ be two components of $B$ such that $|C_1|\leq 7$ and $C_1\neq C$. If there are edges $u_1\hat C$ and $v_1\hat C_2$ such that $u_1, v_1\in V(C_1)$ and there is a Hamiltonian $u_1, v_1$-path in $G[V(C_1)]$, then one can compute in polynomial time a canonical $2$-edge cover $S'$ of $G$ such that every component of $S'$ is 2EC, $S'$ has fewer components than $S$, and $\cost(S')\leq \cost(S)$.
\end{lemma}

\begin{proof}
    Let $F$ be the cycle in $B$ formed by the edges $u_1\hat C, v_1\hat C_2$, and a path in $B$ from $C_2$ to $C$ not going through $C_1$. Note that such a path exists because $B$ is 2VC. Let $P$ be the Hamiltonian $u_1, v_1$-path in $G[V(C_1)]$ (which can be found in polynomial time because $|C_1|\leq 7$). We set $S':= (S\setminus C_1)\cup F\cup P$. In $S'$ there is a large component $C'$ spanning the nodes of all the components of $B$ incident on $F$. Moreover, every such component yields at least $1$ credit, while $C$ yields $2$ credits, so we collect at least $|F|+1$ credits from these components. One has $\cost(S) - \cost(S') \ge |S|-|S'|+|F|+1-\credit(C') = -|F| + 1 + |F| + 1 - 2 = 0$. We remark that $S'$ is canonical and it contains fewer components than $S$, so the lemma follows.
\end{proof}

\begin{lemma}\label{lem:gluingC4LocalC5}
    Let $C_1$ be a component of $B$ that is either a $4$-cycle or a local $5$-cycle. Then, in polynomial time one can compute a canonical $2$-edge cover $S'$ of $G$ such that every component of $S'$ is 2EC, $S'$ has fewer components than $S$, and $\cost(S')\leq \cost(S)$.
\end{lemma}
\begin{proof}
    Apply Lemma~\ref{lem:shortcutC4localC5} with $C_1:=C_1, C_2:=C$ and with $u_2\hat X$ being an arbitrary edge of $B$ incident to $C$. This way one can find in polynomial time a cycle $F$ in $B$ incident to $C$ and to nodes $u_1, v_1$ of $C_1$ such that there is a Hamiltonian $u_1, v_1$-path $P$ in $G[V(C_1)]$. Let $S':= (S\setminus C_1)\cup F\cup P$. In $S'$ there is a large component $C'$ spanning the nodes of all components of $B$ incident on $F$. Also, every such component yields at least $1$ credit, while $C$ yields $2$ credits, so we collect at least $|F|+1$ credits from these components. Thus, one has $\cost(S)-\cost(S') \geq |S| -|S'|+|F|+1 -\credit(C')= -|F| + 1 + |F|+1 -2=0$. We remark that $S'$ is canonical and it contains fewer components than $S$, so the lemma follows.
\end{proof}

\begin{lemma}\label{lem:gluingNonLocalC5}
    Let $C_1$ be a component of $B$ that is a non-local $5$-cycle. Then, one can compute in polynomial time a canonical $2$-edge cover $S'$ of $G$ such that every component of $S'$ is 2EC, $S'$ has fewer components than $S$, and $\cost(S')\leq \cost(S)$.
\end{lemma}

\begin{proof}
    We assume that the conditions of Lemmas~\ref{lem:gluingAdjacent} and~\ref{lem:gluingC4LocalC5} do not hold, otherwise we are done. In particular, we can assume $B$ contains no $4$-cycle nor local $5$-cycle. We consider the following cases:

    \begin{caseanalysis}
    \case{$|V(B)|=2$.} We show by contradiction that this case is not possible. It must be that $B$ contains exactly $2$ components, $C$ and $C_1$. Since $C_1$ is a $5$-cycle, by applying the local $3$-matching lemma (Lemma~\ref{lem:local3matching}) to $C_1$ in $B$, one can find a pair of edges $u_1\hat C, v_1\hat C$ such that $u_1v_1\in E[C_1]$. This is a contradiction to the fact that the conditions of Lemma~\ref{lem:gluingAdjacent} do not hold (with $C_1:=C_1, C_2:=C$).

    \case{$|V(B)|=3$.} Let $C_2$ be the component of $B$ other than $C$ and $C_1$. 

    \subcase{$C_2$ is a $5$-cycle.}\label{case:gluingTwoC5} We show by contradiction that this case is not possible. See Figure~\ref{fig:gluingTwoC5} for an illustration. Apply the local $3$-matching lemma (Lemma~\ref{lem:local3matching}) to $C$ and $B$ to see that there must exist a matching $M$ of size $2$ between $C$ and either $C_1$ or $C_2$. W.l.o.g.\ assume that $M = \{uu_1, vv_1\}$, where $u, v\in V(C)$ and $u_1, v_1\in V(C_1)$. Notice that it must be that $u_1$ and $v_1$ are not adjacent in $C_1$, otherwise the edges $uu_1, vv_1$ imply that the conditions of Lemma~\ref{lem:gluingAdjacent} hold (with $C_1:=C_1, C_2:=C$), a contradiction.  Now, apply the local $3$-matching lemma (Lemma~\ref{lem:local3matching}) to $C_1$ and $B$ to see that there must exist an edge $w_1\hat X$, where $w_1\in V(C_1)\setminus\{u_1, v_1\}$ and $\hat X\in V(B)\setminus \{\hat C_1\}$. But then, since $C_1$ is a $5$-cycle, $w_1\notin\{u_1, v_1\}$ and $u_1, v_1$ are not adjacent in $C_1$, either $u_1$ or $v_1$ is adjacent to $w_1$ in $C_1$. Say w.l.o.g.\ $u_1$ is adjacent to $w_1$ in $C_1$. Then, the edges $uu_1, w_1\hat X$ imply that $C_1$ and $X$ satisfy the conditions of Lemma~\ref{lem:gluingAdjacent} (with $C_1:=C_1, C_2:=X$, $u_1:=v_1$ and $v_1:=w_1$), a contradiction.

    \begin{figure}
        \centering
        \includegraphics[scale=0.65]{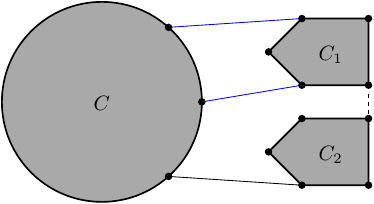}
        \caption{An illustration of Case~\ref{case:gluingTwoC5}. The solid blue edges are the matching $M$, while the dashed edge is the edge $w_1\hat X$. In the figure, one has that $X = C_2$.}
        \label{fig:gluingTwoC5}
    \end{figure}

    \subcase{Otherwise $|C_2|\geq 6$}.\label{case:gluingOtherBlock} Apply Corollary~\ref{cor:cycleSize3} with $C_1:=C_1, C_2:=C$ and with $u_1\hat X$ being an arbitrary edge of $B$ incident to some $u_1\in V(C_1)$. This way one can find in polynomial time a cycle $F$ in $B$ incident to distinct nodes $u_1, v_1$ of $C_1$ and such that $|F| = 3$. 
    Thus, $V(F) = \{\hat C, \hat C_1, \hat C_2\}$. 
    Now, since $|C_2|\geq 6$, $\credit(C_2)\geq 6/4$. Also, since $\credit(C) = 2$ and $\credit(C_1)=5/4$, one has that $\sum_{\hat X\in V(F)}\credit(X) \geq |F| + 7/4$. In the following, only the fact that $F$ is a cycle in $B$ incident to $C$ and to distinct nodes $u_1, v_1$ of $C_1$ such that $\sum_{\hat X\in V(F)}\credit(X) \geq |F| + 7/4$ is needed. This will be useful to avoid duplicated arguments in later cases.
    
    Since $C_1$ is a non-local $5$-cycle, it also belongs to some 2VC block $B'$ of $\hat G_S$ distinct from $B$. Apply the local $3$-matching lemma (Lemma~\ref{lem:local3matching}) to $C_1$ in $B'$ to find an edge $w_1\hat C_1'$ with $\hat C_1'\in V(B')$ and $w_1\in V(C_1)\setminus\{u_1, v_1\}$. Consider now the following subcases:

    \subsubcase{$C'_1$ is a $4$-cycle.}\label{case:gluingOtherBlockC4} We illustrate this case in Figure~\ref{fig:gluingOtherBlockC4}. Apply Lemma~\ref{lem:shortcutC4localC5} with $C_1:=C_1', C_2:=C_1$ and with $u_2\hat X:=w_1\hat X$. This way one can find in polynomial time a cycle $F'$ in $B'$ incident to nodes $w_1, x_1$ in $C_1$ and to nodes $u_1', v_1'$ in $C_1'$ such that there is a Hamiltonian $u_1', v_1'$-path $P$ in $G[V(C_1')]$. Here $x_1$ is a node of $C_1$ distinct from $w_1$. Set $S'':= (S\setminus C_1')\cup F\cup F'\cup P$. In $S''$ there is a large component $C''$ spanning the nodes of all components of $S$ incident on $F$ or $F'$. Notice that $|S''| = |S| + |F| + |F'| - 1$. One has that $\sum_{\hat X\in V(F)}\credit(X) \geq |F| + 7/4$, and every component of $B'$ incident with $F'$ yields at least $1$ credit. Since $B$ and $B'$ are both 2VC blocks of $\hat G_S$, the only component present in both $B$ and $B'$ is $C_1$. Thus, the components of $S$ incident on either $F$ or $F'$ yield at least $|F| + 7/4 + |F'| - 1$ credits. Thus, $\credit(S'') \leq \credit(S) - |F| - |F'| - 3/4 + \credit(C'') = \credit(S) - |F| - |F'| + 5/4$. 

    Observe that $C''\setminus C_1$ contains a $u_1, v_1$-path via $F$ and a $w_1, x_1$-path via $F'$. Since $w_1\notin \{u_1, v_1\}$, $u_1\neq v_1$, and $w_1\neq x_1$, $C''$ and $C_1$ satisfy the conditions of Corollary~\ref{cor:shortcutNonLocalC5} (with $F:=C''$ and $C:=C_1$). Therefore, one can find an edge $e\in C_1$ such that $C''\setminus\{e\}$ is 2EC. Set $S':=S''\setminus\{e\}$. In $S'$ there is a large component $C'$ spanning the nodes of all components of $S$ incident on $F$ or $F'$. Moreover, $\credit(S) = \credit(S'')$, so one has $\cost(S) - \cost(S') = |S| - |S'| + \credit(S) - \credit(S') \geq - |F| - |F'| + 2 + |F| + |F'| - 5/4 > 0$. We remark that $S'$ is canonical and it contains fewer components than $S$, so this subcase follows.

    \begin{figure}
        \centering
        \includegraphics[scale=0.65]{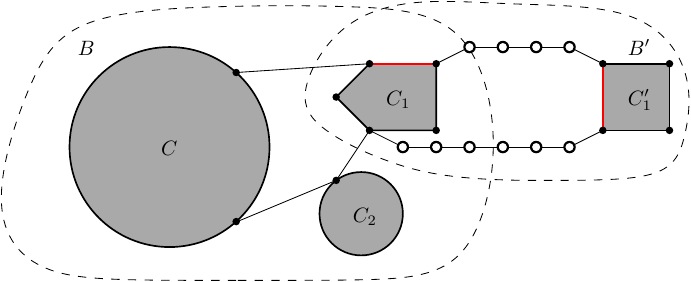}
        \caption{An illustration of Case~\ref{case:gluingOtherBlockC4}. The red edges are removed from $S$. The red edge in $C_1$ is removed through the application of Corollary~\ref{cor:shortcutNonLocalC5}.}
        \label{fig:gluingOtherBlockC4}
    \end{figure}

    \subsubcase{Otherwise $|C_1'|\geq 5$.}
    
    \label{case:gluingOtherBlockOther} We illustrate this case in Figure~\ref{fig:gluingOtherBlockOther}. We have $\credit(C_1')\geq 5/4$. Apply Lemma~\ref{lem:niceCycle} with $C_1:=C_1, C_2:=C_1'$, $u_1\hat X_1:=w_1\hat X$ and $u_2\hat X_2$ being an arbitrary edge of $B$ incident on $C_1'$. This way, one can find in polynomial time, a cycle $F'$ in $B'$ incident on distinct nodes $w_1$ and $x_1$ in $C_1$ and on $C_1'$.
    Set $S'':= S\cup F\cup F'$. In $S''$ there is a large component $C''$ spanning the nodes of all components of $S$ incident on $F$ or $F'$. Notice that $|S''| = |S| + |F| + |F'|$. One has that $\sum_{\hat X\in V(F)}\credit(X) \geq |F| + 7/4$, and every component of $B'$ incident on $F'$ yields at least $1$ credit, while $C_1'$ yields at least $5/4$ credits. Since $B$ and $B'$ are both 2VC blocks of $\hat G_S$, the only component present in both $B$ and $B'$ is $C_1$. Thus, the components of $S$ incident on either $F$ or $F'$ yield at least $|F| + 7/4 + (|F'| - 2) + 5/4 = |F| + |F'| +1$ credits. Thus, $\credit(S'') \leq \credit(S) - |F| - |F'| - 1 + \credit(C'') = \credit(S) - |F| - |F'| + 1$. 

    Observe that $C''\setminus C_1$ contains a $u_1, v_1$-path via $F$ and a $w_1, x_1$-path via $F'$. Since $w_1\notin \{u_1, v_1\}$, $u_1\neq v_1$, and $w_1\neq x_1$, $C''$ and $C_1$ satisfy the conditions of Corollary~\ref{cor:shortcutNonLocalC5} (with $F:=C''$ and $C:=C_1$). Thus, we can find an edge $e\in C_1$ such that $C''\setminus\{e\}$ is 2EC. Set $S':=S''\setminus\{e\}$. In $S'$ there is a large component $C'$ spanning the nodes of all components of $S$ incident on $F$ or $F'$. Moreover, $\credit(S) = \credit(S'')$, so one has $\cost(S) - \cost(S') = |S| - |S'| + \credit(S) - \credit(S') \geq - |F| - |F'| + 1 + |F| + |F'| - 1 = 0$. We remark that $S'$ is canonical and it contains fewer components than $S$, so this subcase follows.

    \begin{figure}
        \centering
        \includegraphics[scale=0.65]{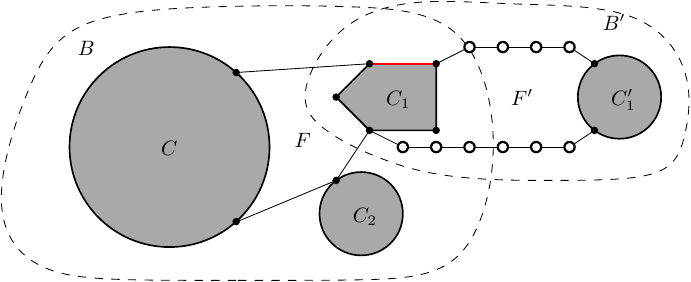}
        \caption{An illustration of Case~\ref{case:gluingOtherBlockOther}. The red edge in $C_1$ is removed through the application of Corollary~\ref{cor:shortcutNonLocalC5}.}
        \label{fig:gluingOtherBlockOther}
    \end{figure}

    \case{$|V(B)|\geq 4$.} Since $B$ is 2VC and $|V(B)| > 3$ one can find in polynomial time a cycle $F$ in $B$ containing $\hat C$ and $\hat C_1$ such that $|F|\geq 3$. We first show that we can assume that $|F|\geq 4$. Assume this is not the case, i.e., $|F|=3$. Since $|V(B)|\geq 4$, there exists an edge $\hat X\hat Y\in B$ such that $\hat X\in V(F), \hat Y\in V(B)\setminus V(F)$. Since $B$ is 2VC, we can find in polynomial time a path $P$ from $\hat Y$ to some node $\hat X'\in V(F)\setminus\{\hat X\}$ not going through $\hat X$. The cycle $(F\setminus\{\hat X\hat X'\})\cup P\cup \{\hat X\hat Y\}$ has length at least $4$ and contains the nodes $C$ and $C_1$ (because it contains every node in $V(F)$), so it is our desired cycle. 
        
    We now show that we can assume that $|F| = 4$ and every component of $B$ incident on $F$ other than $C$ is a non-local $5$-cycle. Indeed, every component of $B$ incident on $F$ yields at least $5/4$ credits if it is a $5$-cycle, and at least $6/4$ otherwise (while $C$ yields $2$ credits). Therefore, if $|F|\geq 5$ or $|F|=4$ and contains at least one component with at least $6$ edges, one has that $\sum_{\hat X\in V(F)}\credit(X)\geq |F| + 2$. Hence, in $S':=S\cup F$ there is a large component $C'$ spanning the nodes of all components of $B$ incident on $F$, and $\cost(S)-\cost(S') = |S|-|S'|+\credit(S)-\credit(S')\geq -|F| + |F| + 2 -\credit(C') = 0$. Clearly, $S'$ is canonical and has fewer components than $S$, so it satisfies the claim of the lemma.

    By our above assumptions, we can assume w.l.o.g.\ that $F=\{\hat C\hat C_2, \hat C\hat C_3, \hat C_1\hat C_2, \hat C_1\hat C_3\}$, where $C_2$ and $C_3$ are components of $B$ that are non-local $5$-cycles. We will show how to find a cycle $F'$ in $B$ such that $F'$ is incident to $C$ and to distinct nodes $u_1, v_1$ of $C_1$, and $|F'|\geq 4$. This is enough to prove the claim of the lemma, because if $|F'|\geq 4$, since Lemma~\ref{lem:gluingC4LocalC5} does not hold, every component of $B$ incident with $F'$ contains at least $5$ edges and thus yields at least $5/4$ credits, while $C$ yields $2$ credits, then $\sum_{\hat X\in V(F)}\credit(X) \geq |F| + 7/4$. Therefore, we can apply the exact same arguments as in Case~\ref{case:gluingOtherBlock} above.

    Assume $F$ does not satisfy these properties. Since $|F|=4$ and is incident to $C$ and $C_1$, the edges $\hat C_1\hat C_2$ and $\hat C_1\hat C_3$ must be both incident to the same node in $C_1$, say $u_1$. By applying the local $3$-matching lemma (Lemma~\ref{lem:local3matching}) to $C_1$ in $B$, and since $F$ is incident to $u_1$, there always exists a  matching $\{u_1\hat X_1, u_2\hat X_2, u_3\hat X_3\}$ in $B$ such that $u_2, u_3\in V(C_1)$ and $\hat X_i\in V(B)\setminus V(C_1)$, for $i\in \{1, 2, 3\}$. If, for some $i\in \{2, 3\}$ $X_i = C_2$, then the cycle $F':=(F\setminus\{u_1\hat C_2\})\cup \{u_i\hat C_2\}$ has the desired properties. A symmetric argument works when $X_i = C_3$ for some $i\in\{2, 3\}$.

    \begin{figure}
        \centering
        \includegraphics[scale=0.65]{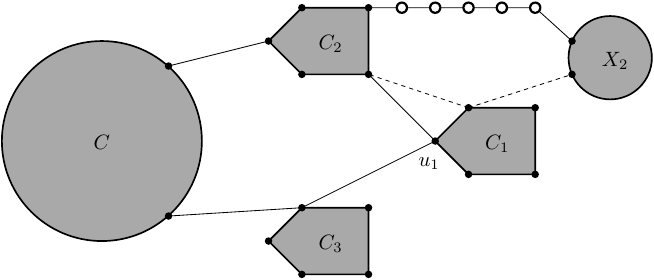}
        \caption{We illustrate how to compute a cycle of size at least $4$ that is incident to distinct nodes of $C_1$. In the figure, any of the dashed edges yields a cycle of the desired type.}
        \label{fig:gluingLast}
    \end{figure}

    Suppose now $X_2 = X_3 = C$. Then, there is a matching of size $2$ between $C$ and $C_1$, and by the same arguments as in Case~\ref{case:gluingTwoC5}, this leads to a contradiction to the fact that Lemma~\ref{lem:gluingAdjacent} does not hold. Thus, we can assume w.l.o.g.\ that $X_2\notin \{C, C_1, C_2, C_3\}$. Since $B$ is 2VC, one can find in polynomial time a path $P$ in $B$ from $X_2$ to a node in $V(F)$ not going through $C_1$. If $P$ has $C_2$ as an endpoint, the cycle $F':=\{\hat C\hat C_2, \hat C\hat C_3, u_1\hat C_3, u_2\hat X_2\}\cup P$ has the desired properties. A symmetric argument works if $P$ has $C_3$ as an endpoint. Finally, if $P$ has $C$ as an endpoint, the cycle $F':=\{\hat C\hat C_2, u_1\hat C_2, u_2\hat X_2\}\cup P$ has the desired properties (see Figure~\ref{fig:gluingLast}).
    \end{caseanalysis}
\end{proof}

\begin{lemma}\label{lem:gluingC6C7}
    Let $|V(B)|=2$ and let $C_1$ be the component of $B$ distinct from $C$. If $|C_1|\geq 6$, then one can compute in polynomial time a canonical $2$-edge cover $S'$ of $G$ such that every component of $S'$ is 2EC, $S'$ has fewer components than $S$, and $\cost(S')\leq \cost(S)$.
\end{lemma}

\begin{proof}
    If $|C_1|\geq 8$, then $\credit(C_1) = 2$. In this case we simply set $S':=S\cup\{e_1, e_2\}$, where $e_1$ and $e_2$ are two distinct edges between $V(C)$ and $V(C_1)$, which are guaranteed by the local $3$-matching lemma (Lemma~\ref{lem:local3matching}) applied to $C$ and $B$. In $S'$ there is a large 2EC component $C'$ spanning the nodes of $C$ and $C_1$. One has $\cost(S)-\cost(S') = |S|-|S'|+\credit(C)+\credit(C_1)-\credit(C') = -2 + 2 + 2 - 2 = 0$.

    From now on we assume that $C_1$ is a $6$-cycle or a $7$-cycle, as otherwise we will be immediately done from the previous lemmas. Assume first that there are edges $u_1\hat{C}$ and $v_1\hat{C}$, such that there exists a $u_1, v_1$-Hamiltonian path in $G[V(C_1)]$. By this assumption, $C_1$ satisfies the conditions of Lemma~\ref{lem:gluingAdjacent} (with $C_2:=C$), so we can construct the desired $S'$. Therefore, we can assume that this is not the case.

    Apply the local $3$-matching lemma (Lemma~\ref{lem:local3matching}) to $C_1$ and $B$ to find a matching $\{u_1\hat C, v_1\hat C, w_1\hat C\}$ between $\{u_1, v_1, w_1\}\subseteq V(C_1)$ and $V(C)$ in $G$. By the above argument, we can assume w.l.o.g.\ that $u_1, v_1$, and $w_1$ are pairwise non-adjacent in $C_1$, otherwise, there would be a Hamiltonian path in $G[V(C_1)]$ between them, a contradiction to the assumption above. Moreover, it must be that there is no Hamiltonian path in $G[V(C_1)]$ between any pair in $\{u_1, v_1, w_1\}$. By the Hamiltonian pairs lemma (Lemma~\ref{lem:hamiltonianPairs}), there must exist an edge $x_1\hat X$ in $G$ such that $x_1\in V(C_1)\setminus \{u_1, v_1, w_1\}$ and $\hat X\in V(B')$, where $B'$ is a 2VC block of $\hat G_S$. Again, by the above argument, it must be that $B'\neq B$. Now, Since $C_1$ is a $6$-cycle or a $7$-cycle and $u_1, v_1$ and $w_1$ are pairwise non-adjacent, we can assume w.l.o.g.\ that $x_1$ is part of the shortest $u_1, v_1$-path in $C_1$. Let $F:=\{u_1\hat C, v_1\hat C\}$ and let $C_2$ be the component of $B'$ with the least number of edges and distinct from $C_1$.
    
    We now consider two cases:

    \begin{caseanalysis}
    \case{$C_2$ is a $4$-cycle.}\label{case:gluingC6C7:C4} This case is illustrated in Figure~\ref{fig:gluingC6C7:C4}. Apply Lemma~\ref{lem:shortcutC4localC5} to find a cycle $F'$ in $B'$ incident to distinct nodes $x_1, y_1$ in $C_1$ and to distinct nodes $u_2, v_2$ in $C_2$ such that there is a Hamiltonian $u_2, v_2$-path $P$ in $G[V(C_2)]$. 
    Set $S'':= (S\setminus C_2)\cup F\cup F'\cup P$. In $S''$ there is a large component $C''$ spanning the nodes of all components of $S$ incident on $F$ or $F'$. Every such component yields at least $1$ credit, while $C$ yields $2$ credits and $C_1$ yields at least $6/4$ credits. Since $B$ and $B'$ are 2VC blocks of $\hat G_S$, $C_1$ is the only component present in both $F$ and $F'$, $\credit(S'') \leq \credit(S) -(|F|-2)- (|F'|-1)-\credit(C)-\credit(C_1)+\credit(C'')=\credit(S) -(|F|-2)- (|F'|-1)-2-6/4+2<\credit(S) -|F| - |F'| + 2$. Notice that $|S''|=|S|+|F|+|F'|-1$.
    
    Observe that $C''\setminus C_1$ contains a $u_1, v_1$-path and a $x_1, y_1$-path. Since $x_1$ is in the shortest $u_1, v_1$-path in $C_1$, $u_1\neq v_1$, and $x_1\neq y_1$, $C''$ and $C_1$ satisfy the conditions of Corollary~\ref{cor:shortcutNonLocalC6C7}. Therefore, we can find an edge $e\in E[C_1]$ such that $C''\setminus\{e\}$ is 2EC. Set $S':=S''\setminus\{e\}$, where we remark that $\credit(S')=\credit(S'')$. One has $\cost(S) - \cost(S') = |S| - |S'| + \credit(S) - \credit(S') > -~|F| - |F'| + 2 + |F| + |F'| - 2 = 0$.

    \begin{figure}[t]
\centering
 \begin{subfigure}[b]{0.45\textwidth}
        \centering
        \resizebox{\linewidth}{!}{
\includegraphics[scale=0.65]{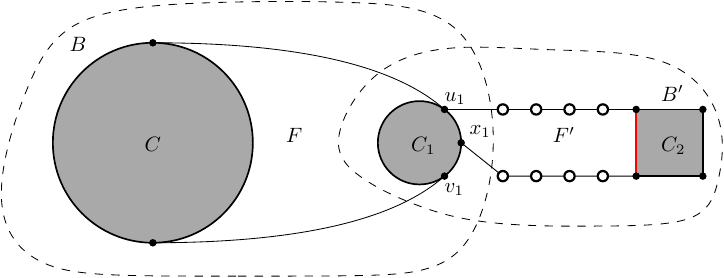}
}
        \caption{Case~\ref{case:gluingC6C7:C4}.}
        \label{fig:gluingC6C7:C4}
    \end{subfigure}
     \hspace{0.5cm}
  \begin{subfigure}[b]{0.45\textwidth}
    \centering
        \resizebox{\linewidth}{!}{
\includegraphics[scale=0.65]{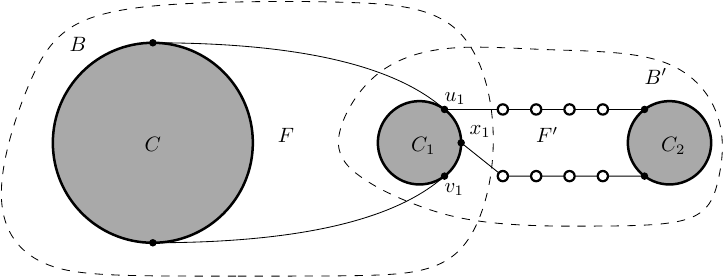}
}
		 \caption{Case~\ref{case:gluingC6C7:general}. }
        \label{fig:gluingC6C7:general}
    \end{subfigure}
\caption{{\bf (a)} Illustration of Case~\ref{case:gluingC6C7:C4}. The red edge is removed from $S$. An additional edge of $C_1$ is removed from $S$ through the application of Corallary~\ref{cor:shortcutNonLocalC6C7}.
{\bf (b)} An illustration of the construction for Case~\ref{case:gluingC6C7:general}. In most cases this construction will be enough, although additional steps are needed in a few subcases. An edge from $C_1$ is removed from $S$ through the application of Corollary~\ref{cor:shortcutNonLocalC6C7}.} 
\label{fig:gluingC6C7:combined2}
\end{figure}

    \case{Otherwise $|C_2|\geq 5$}.\label{case:gluingC6C7:general} Apply Corollary~\ref{cor:cycleSize3} to find a cycle $F'$ in $B'$ incident to distinct nodes $x_1, y_1$ in $C_1$ and to $C_2$ such that $|F'|\geq\min\{3, |V(B')|\}$. 
    Set $S'':= S\cup F\cup F'$. In $S''$ there is a large component $C''$ spanning the nodes of all components of $S$ incident on $F$ or $F'$. Since Case~\ref{case:gluingC6C7:C4} does not hold, it must be that every component of $B'$ has at least $5$ edges. Thus, every component of $B'$ incident with $F'$ yields at least $5/4$ credits. Since $B$ and $B'$ are 2VC blocks of $\hat G_S$, $C_1$ is the only component of $S$ present in both $F$ and $F'$. Since $V(F) = \{\hat C, \hat C_1\}$, it holds that:
    \begin{align*}
        \credit(S) - \credit(S'') &\geq \credit(C) + \credit(C_1) + \credit(C_2) + \frac{5}{4}(|F'|-2)-\credit(C'') \\
        &= \credit(C_1) + \credit(C_2) + \frac{5}{4}(|F'|-2).
    \end{align*}

    Observe that $C''\setminus C_1$ contains a $u_1, v_1$-path via $F$ and a $x_1, y_1$-path via $F'$ (see Figure~\ref{fig:gluingC6C7:general}). Since $x_1$ is on the shortest $u_1, v_1$-path in $C_1$, $u_1\neq v_1$, and $x_1\neq y_1$, $C''$ and $C_1$ satisfy the conditions of Corollary~\ref{cor:shortcutNonLocalC6C7} (with $F:=C''$ and $C:=C_1$). Therefore, we can find in polynomial time an edge $e\in C_1$ such that $C''\setminus \{e\}$ is 2EC. Set $S':=S''\setminus\{e\}$, where we note that $\credit(S')=\credit(S'')$. Using that $|F|=2$ one has:
    \begin{align*}
        \cost(S) - \cost(S') &= |S| - |S'| + \credit(S) -\credit(S') \\
        &\geq |S| - |S'| + \credit(C_1) + \credit(C_2) + \frac{5}{4}(|F'|-2) \\
        &= -|F| -|F'| + 1 + \credit(C_1) + \credit(C_2) + \frac{5}{4}|F'|- \frac{10}{4} \\
        &=\credit(C_1) + \credit(C_2) + \frac{1}{4}|F'| - \frac{14}{4}.
    \end{align*}
    Consider now the following subcases:
    
    \subcase{$|C_1|= 7$.} Then one has that $\credit(C_1)\geq 7/4$ and  
    $\credit(C_2)\geq 5/4$. Thus, $\cost(S)-\cost(S')\geq \credit(C_1) + \credit(C_2) + |F'|/4 - 14/4\geq 7/4 + 5/4 + 2/4 - 14/4 = 0$. The last inequality follows from the fact that $|F'|\geq 2$.
    
    \subcase{$|C_2|\geq 6$.} Then one has that $\credit(C_2)\geq 6/4$, and by assumption of the lemma, $\credit(C_1)\geq 6/4$. Thus, $\cost(S)-\cost(S')\geq \credit(C_1) + \credit(C_2) + |F'|/4 - 14/4\geq 6/4 + 6/4 + 2/4 - 14/4 = 0$. The last inequality follows from the fact that $|F'|\geq 2$.
    
    \subcase{$|F'|\geq 3$.} By assumption of the lemma, $\credit(C_1)\geq 6/4$, and 
    $\credit(C_2)\geq 5/4$. Thus, $\cost(S)-\cost(S')\geq \credit(C_1) + \credit(C_2) + |F'|/4 - 14/4\geq 6/4 + 5/4 + 3/4 - 14/4 = 0$. The last inequality follows from this case assumption.

    \subcase{Otherwise.} By exclusion of the previous cases and the assumption of the lemma, it must be that $|C_1|=6, |C_2|=5$ and $|F'|=2$. Since $|F'|\geq \min\{3, |V(B')|\}$, it must be that $V(B')=\{\hat C_1, \hat C_2\}$. In this case we show how to compute an alternate $S'$. Let $C_1 = a_1a_2a_3a_4a_5a_6$, and assume w.l.o.g.\ $u_1=a_1, v_1=a_3, w_1=a_5$, and $x_1=a_2$. We recall that $x_1$ is adjacent to some component-node in $V(B')\setminus\{\hat C_1\}$, and since $V(B')=\{\hat C_1, \hat C_2\}$, it must be that there exists an edge $x_1b_1$ for some $b_1\in V(C_2)$.

    \subsubcase{There exists an edge $a_ib_2$ for some $b_2\in V(C_2)$ and $i\in\{2, 4, 6\}$ such that $a_ib_2\neq a_2b_1$.}\label{case:gluingC6C7:degenerate} 
    If $i=4$ then let $S':=(S\setminus\{a_1a_2, a_3a_4\})\cup \{a_1\hat C, a_3\hat C, a_2b_1, a_4b_2\}$. In $S'$ there is a large component $C'$ spanning the nodes of $C, C_1$, and $C_2$. Indeed, $C'$ consists of the components $C$ and $C_2$ joined by the paths $\hat C a_1a_6a_5a_4\hat C_2$ and $\hat C a_3a_2\hat C_2$ (see Figure~\ref{fig:gluingC6C7:degenerate1}). One has $\cost(S) - \cost(S') = |S|-|S'|+\credit(C)+\credit(C_1)+\credit(C_2)-\credit(C')=-2+2+6/4+5/4-2>0$. 
    The case where $i=6$ is symmetric. Assume now $i=2$, so since $a_ib_2\neq a_2b_1$ it must be that $b_1\neq b_2$. By the local $3$-matching lemma (Lemma~\ref{lem:local3matching}) and since $V(B')=\{\hat C_1, \hat C_2\}$, there is a matching $M$ of size $3$ between $V(C_1)$ and $V(C_2)$.

    \begin{figure}
        \centering
        \includegraphics[scale=0.65]{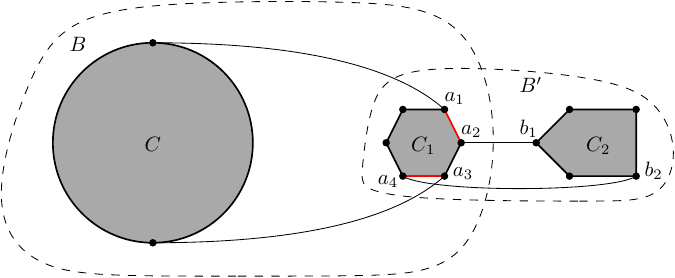}
        \caption{Illustration of Case~\ref{case:gluingC6C7:degenerate}, when $i=4$. The red edges are removed from $S$.}
        \label{fig:gluingC6C7:degenerate1}
    \end{figure}
    \begin{claim}\label{claim:gluingC6C7:degenerate}
        There exists a matching $\{a_2b, ab'\}$ such that $a\in V(C_1)$ and $bb'\in E[C_2]$.
    \end{claim}

    \begin{proof} 
    Since $C_2$ is a $5$-cycle, $|M|=3$, and $b_1\neq b_2$, there is an edge $e\in M$ incident to a node in $C_2$ that is adjacent to either $b_1$ or $b_2$. Let $b'$ be the endpoint of $e$ in $C_2$ and $b\in\{b_1, b_2\}$ be the node of $C_2$ such that $bb'\in E[C_2]$. If $e$ is not incident on $a_2$ in $C_1$ then $\{a_2b, e\}$ is the desired matching. Assume this is not the case, so that there exist edges $a_2b$ and $a_2b' = e$.
    
    Since $|M|=3$, $b\neq b'$, and $C_2$ is a $5$-cycle, there exist an edge $e'\in M\setminus\{e\}$ incident on a node adjacent in $C_2$ to either $b$ or $b'$. Assume w.l.o.g.\ $e'$ is incident on a node in $C_2$ adjacent to $b$. Since $M$ is a matching and $e$ is incident to $a_2$ in $C_1$, $e'$ is not incident on $a_2$ in $C_1$ and thus $\{e', a_2b\}$ is the desired matching. The claim follows.
    \end{proof}

    Let $\{a_2b, ab'\}$ be the matching obtained from Claim~\ref{claim:gluingC6C7:degenerate}. Let $S'':=(S\setminus \{bb'\})\cup \{a_1\hat C, a_3\hat C, a_2b, ab'\}$. In $S''$ there is a large 2EC component $C''$ spanning the nodes of $C, C_1$, and $C_2$ (see Figure~\ref{fig:gluingC6C7:degenerate2}). Notice that $C''\setminus C_1$ contains a $a_1, a_3$-path and a $a_2, a$-path. Since $a_2$ is in the shortest $a_1, a_3$-path in $C_1$ and $a_2\neq a$, $C''$ and $C_1$ satisfy the conditions of Corollary~\ref{cor:shortcutNonLocalC6C7} (with $F:=C''$ and $C:=C_1$). Therefore, in polynomial time we can find an edge $e\in C_1$ such that $C''\setminus \{e\}$ is 2EC. Set $S':=S''\setminus\{e\}$ and let $C'$ be the large 2EC component of $S'$ spanning the nodes of $C, C_1$, and $C_2$. One has $\cost(S)-\cost(S') = |S| - |S'| + \credit(C)+\credit(C_1)+\credit(C_2)-\credit(C') = -2 + 2+ 6/4+5/4-2>0$. 

\begin{figure}[t]
\centering
 \begin{subfigure}[b]{0.45\textwidth}
        \centering
        \resizebox{\linewidth}{!}{
\includegraphics[scale=0.65]{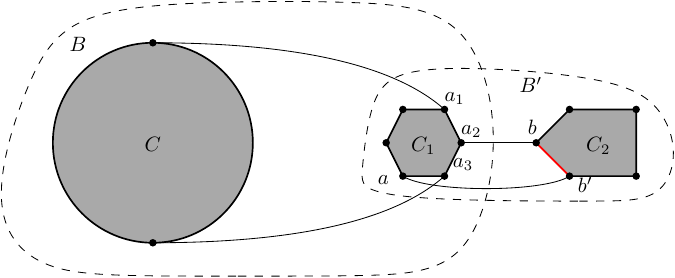}

}
        \caption{Case~\ref{case:gluingC6C7:degenerate}, when $i=4$.}
        \label{fig:gluingC6C7:degenerate2}
    \end{subfigure}
    \hspace{1cm}
  \begin{subfigure}[b]{0.24\textwidth}
    \centering
        \resizebox{\linewidth}{!}{

\includegraphics[scale=0.65]{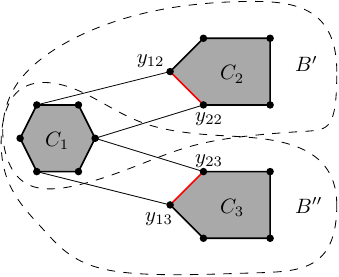}

}
		 \caption{Case~\ref{case:gluingC6C7:otherwise}. }
        \label{fig:gluingC6C7:otherwise}
    \end{subfigure}
\caption{{\bf (a)} Illustration of Case~\ref{case:gluingC6C7:degenerate}, when $i=4$. The red edge is removed from $S$. An additional edge of $C_1$ is removed from $S$ through the application of Corollary~\ref{cor:shortcutNonLocalC6C7}.
{\bf (b)} Illustration of Case~\ref{case:gluingC6C7:otherwise}. The red edges are removed from $S$. Notice that, in this case, the components are not merged together with the component $C$.} 
\label{fig:gluingC6C7:combined2}
\end{figure} 
    \subsubcase{Otherwise.}\label{case:gluingC6C7:otherwise} This case is illustrated in Figure~\ref{fig:gluingC6C7:otherwise}. Recall that there are no Hamiltonian paths in $G[V(C_1)]$ between any pair in $\{u_1, v_1, w_1\}=\{a_1, a_3, a_5\}$. Moreover, by exclusion of Case~\ref{case:gluingC6C7:degenerate}, there is only one edge between $\{a_2, a_4, a_6\}$ and $V(C_2)$. Thus, by Lemma~\ref{lem:6-cycleContractibility}, there must exist an edge $a_i\hat C_3, i\in\{2, 4, 6\}$, with $C_3\neq C_2$. Since $a_i\notin \{u_1, v_1, w_1\}$, by symmetry with $C_2$, it must be that $\hat C_3\in V(B'')$, where $B''$ is a 2VC block of $\hat G_S$ distinct from $B$ and $B'$ such that $V(B'')=\{\hat C_1, \hat C_3\}$ and $C_3$ is a $5$-cycle. 

    By the local $3$-Matching lemma (Lemma~\ref{lem:local3matching}) and since $V(B')=\{\hat C_1, \hat C_2\}$, $V(B'')=\{\hat C_1, \hat C_3\}$, there is a matching of size $3$ between $V(C_1)$ and $V(C_2)$ and between $V(C_1)$ and $V(C_3)$. Since $C_\ell$ is a $5$-cycle, for $\ell\in\{2, 3\}$, there are edges $y_{1\ell}\hat C_1$ and $y_{2\ell}\hat C_1$ such that $y_{1\ell}y_{2\ell}\in E[C_\ell]$, for $\ell\in\{2, 3\}$. Set $S':=(S\setminus\{y_{12}y_{22}, y_{13}y_{23}\})$ $\cup \{y_{12}\hat C_1, y_{22}\hat C_1, y_{13}\hat C_1, y_{23}\hat C_1\}$. In $S'$ there is a large 2EC component $C'$ spanning the nodes of $C_1, C_2$ and $C_3$. One has $\cost(S)-\cost(S') = |S|-|S'|+\credit(C_1)+\credit(C_2)+\credit(C_3)-\credit(C') = - 2 + 6/4 + 5/4 + 5/4 -2 = 0$.
    \end{caseanalysis}

    Notice that in every case, the computed $S'$ is canonical.
    \end{proof}

Finally, we prove the main gluing lemma.
    
\begin{proof}[Proof of Lemma~\ref{lem:gluingSimplified}]
We can assume that the conditions of Lemmas~\ref{lem:gluingC4LocalC5},~\ref{lem:gluingNonLocalC5}, and~\ref{lem:gluingC6C7} do not hold; otherwise we are done. Since the conditions of Lemmas~\ref{lem:gluingC4LocalC5} and~\ref{lem:gluingNonLocalC5} do not hold, it must be that every component of $B$ has at least $6$ edges. Since the condition of Lemma~\ref{lem:gluingC6C7} does not hold, it must be that $|V(B)|\geq 3$. Then, since $B$ is 2VC, one can find in polynomial time a cycle $F$ in $B$ containing $\hat C$ and such that $|F|\geq 3$. Set $S':=S\cup F$. In $S'$ there is a large 2EC component $C'$ spanning the nodes of every component of $B$ incident on $F$. Every such component has at least $6$ edges, and thus yields at least $6/4$ credits, while $C$ yields $2$ credits. Therefore, we collect at least $6/4(|F|-1)+2$ credits from those components. One has $\cost(S)-\cost(S')=|S|-|S'|+\credit(S)-\credit(S')\geq -|F|+ 6/4(|F|-1)+2-\credit(C')\geq 0$, where the last inequality follows from the fact that $|F|\geq 3$. Since $S'$ is canonical and it contains strictly less connected components than $S$, 
and the lemma follows.
\end{proof}

\appendix

\section*{Appendix}
\section{A Reduction to Structured Graphs}
\label{sec:reductionStructured}

In this section, we show that we can reduce our original instance to a 2ECSS instance with $\alpha$-structured graphs. This is an improved version of the preprocessing step in~\cite{GGJ23soda}, where the authors presented a similar reduction, but lost a factor of $1+\eps$ in the approximation ratio. We show how to improve their preprocessing so that the reduction is approximation-preserving for every $\alpha\geq 6/5$, up to a small additive constant.

Since we will consider different instances of 2ECSS, we use $\OPT(G)$ to denote a minimum 2EC spanning subgraph of $G$, and $\opt(G) = |\OPT(G)|$. Given $W\subseteq V(G)$ we let $G|W$ denote the multi-graph obtained by contracting $W$ into a single node $w$. For a subgraph $H$ of $G$, we use $G|H$ as a shortcut for $G|V(H)$. We use the following well-known facts and lemmas:

\begin{fact}\label{fact:contract}
    Let G be a 2EC graph and $W\subseteq V(G)$. Then $G|W$ is 2EC.
\end{fact}
\begin{fact}\label{fact:uncontract}
    Let $H$ be a 2EC subgraph of a 2EC graph $G$, and $S$ be 2EC spanning subgraph of $G|H$. Then $S\cup H$ is a 2EC spanning subgraph of $G$, where we interpret $S$ as edges of $G$ after decontracting $V(H)$. 
\end{fact}

\begin{lemma}[\cite{GGJ23soda, HVV19}]\label{lem:parallel}
    Let G be a 2VC multigraph with $|V (G)| \geq 3$. Then, there exists a minimum size 2EC spanning subgraph of G that is simple (i.e. contains no self-loops nor parallel edges).
\end{lemma}

\begin{lemma}[\cite{GGJ23soda, HVV19, SV14}]\label{lem:oneCut}
    Let $G$ be a 2EC graph, $v$ be a $1$-vertex cut of $G$, and $(V_1, V_2), V_1\neq \emptyset\neq V_2$, be a partition of $V\setminus\{v\}$ such that there are no edges between $V_1$ and $V_2$. Then $\opt(G) = \opt(G_1) + \opt(G_2)$, where $G_i := G[V_i\cup\{v\}]$.
\end{lemma}

\begin{lemma}[\cite{GGJ23soda}]\label{lem:irrelevantEdge}
    Let $e = uv$ be an irrelevant edge of a 2VC simple graph $G = (V,E)$. Then there exists a minimum-size 2EC spanning subgraph of $G$ not containing $e$.
\end{lemma}

We will also need the following technical definitions and lemmas from~\cite{GGJ23soda}:

\begin{definition}[\cite{GGJ23soda}]
    Consider any graph $H$ and two nodes $u, v\in V(H)$. Let $H'$ be obtained by contracting each 2EC block $C$ of $H$ into a single super-node $C$, and let $C(u)$ and $C(v)$ be the super-nodes corresponding to $u$ and $v$, respectively. Then $H$ is:
    \begin{enumerate}
        \item of type $A$ (w.r.t.\ $\{u, v\}$) if $H'$ consists of a single super-node $C(u)=C(v)$;
        \item of type $B$ (w.r.t.\ $\{u, v\}$) if $H'$ is a $C(u), C(v)$-path (of length at least $1$);
        \item of type $C$ (w.r.t.\ $\{u, v\}$) if $H$ consists of two isolated super-nodes $C(u)$ and $C(v)$.  
    \end{enumerate}
\end{definition}

\begin{lemma}[\cite{GGJ23soda}]\label{lem:redTypes}
    Let $G$ be a simple 2VC graph without irrelevant edges, $\{u, v\}$ be a $2$-vertex cut of $G$, $V_1\neq \emptyset\neq V_2$, be a partition of $V(G)\setminus\{u, v\}$ such that there are no edges between $V_1$ and $V_2$, and $H$ be a 2EC spanning subgraph of $G$. Define $G_i = G[V_i\cup\{u, v\}]$ and $H_i = E(G_i)\cap H$ for $i\in \{1, 2\}$. The following claims hold:
    \begin{enumerate}
        \item Both $H_1$ and $H_2$ are of type $A$, $B$, or $C$ (w.r.t.\ $\{u, v\}$). Furthermore, if one of the two is of type $C$, then the other must be of type $A$.
        \item If $H_i$ is of type $C$, then there exists an edge $f\in E(G_i)$ such that $H_i\cup \{f\}$ is of type $B$. As a consequence, there exists a 2EC spanning subgraph $H'$ of $G$, such that  $E(G_i)\cap H'$ is of type $A$ or $B$, for $i\in\{1, 2\}$.\label{lem:redTypes:edge}
    \end{enumerate}
\end{lemma}

It is not difficult to see that the reverse condition is also true:

\begin{lemma}\label{lem:connectRedTypes}
    Let $G$ be a simple 2VC graph without irrelevant edges, $\{u, v\}$ be a $2$-vertex cut of $G$, $V_1\neq \emptyset\neq V_2$, be a partition of $V(G)\setminus\{u, v\}$ such that there are no edges between $V_1$ and $V_2$. Define $G_i = G[V_i\cup\{u, v\}]$. Let $H_1, H_2$ be spanning subgraphs of $G_1, G_2$. If both $H_1$ and $H_2$ are of type $A$ or $B$ w.r.t.\ $\{u,v\}$, then $H_1\cup H_2$ is a 2EC spanning subgraph of $G$. Moreover, if $H_1$ is of type $A$ and $H_2$ is of type $C$ w.r.t.\ $\{u,v\}$, then $H_1\cup H_2$ is a 2EC spanning subgraph of $G$.
\end{lemma}
\begin{proof}
    First note if $H_1\cup H_2$ contains a bridge $e$, then $e$ must be a bridge in $H_1$ or $H_2$. To see this, note that since there are no edges between $V_1$ and $V_2$, $e$ belongs to either $H_1$ or $H_2$, so it must separate nodes of either $H_1$ or $H_2$. Say w.l.o.g.\ $e$ is a bridge in $H_1$, so that $H_1$ is of type $B$ and thus $H_2$ is of type $A$ or $B$ by assumption of the lemma. Removing $e$ splits $H_1$ in $2$ components, one containing $u$ and one containing $v$. But in $H_1\cup H_2$, there is a path between those components through $H_2$, a contradiction to the fact that $e$ is a bridge in $H_1\cup H_2$. 
\end{proof}

\begin{algorithm}
\caption{Algorithm $\RED$ that reduces from arbitrary to structured instances of 2ECSS. Here $\ALG$ is an algorithm for $\alpha$-structured instances that returns a solution of cost at most $\alpha\cdot\opt(G) - 2$ for $\alpha \geq 6/5$.}
\begin{algorithmic}[1]
    \If{$|V(G)| \leq \max\{\frac{4}{\alpha - 1}, 5\}$} 
    \State Compute $\OPT(G)$ and \textbf{return} $\OPT(G)$.\label{line:baseCaseBruteForce}
    \EndIf

    \If{$G$ has a $1$-vertex cut $v$} 
    \State Let $(V_1, V_2)$ be a partition of $V\setminus \{v\}$ such that $V_1\neq \emptyset\neq V_2$ and there are no edges between $V_1$ and $V_2$.
    \State \textbf{return} $\RED(G[V_1]\cup\{v\})\cup \RED(G[V_2]\cup\{v\})$.\label{line:oneCut}
    \EndIf

    \If{$G$ contains an edge $e$ such that $e$ is a loop or a parallel edge} 
    \State \textbf{return} $\RED(G - e)$.\label{line:parallel}
    \EndIf

    \If{$G$ contains an edge $e$ such that $e$ is an irrelevant edge} 
    \State \textbf{return} $\RED(G - e)$.\label{line:irrelevant}
    \EndIf
    
    \If{$G$ contains an $\alpha$-contractible subgraph $H$ with $|V(H)|\leq \frac{2}{\alpha - 1}$} 
    \State \textbf{return} $H\cup\RED(G|H)$.\label{line:contract}
    \EndIf

    \If{$G$ contains a non-isolating $2$-vertex cut $\{u, v\}$}\label{line:hardCase}
    \State Let $(V_1, V_2)$ be a partition of $V\setminus\{u, v\}$ such that $2\leq |V_1|\leq |V_2|$ and there are no edges between $V_1$ and $V_2$.\label{line:partition}
    \State Let $G_i := G[V_i\cup\{u, v\}]$, for $i\in\{1, 2\}$.\label{line:g1g2}
    \If{$|V(G_2)| \leq 4/(\alpha -1)$} compute $\OPT(G)$  and \textbf{return} $\OPT(G)$.\EndIf\label{line:bruteForce2}
    \If{$|V(G_1)| > 2/(\alpha -1)$}\label{line:bothBigCase}
    \State Let $G_i' := G_i|\{u, v\}$, for $i\in\{1, 2\}$.
    \State Let $S_i := \RED(G_i')$, for $i\in\{1, 2\}$.
    \State\textbf{return} $S:=S_1\cup S_2\cup F$, where $F\subseteq E(G), |F|\leq 2$ is a minimum set of edges such that $S$ is 2EC.\label{line:bothBigCase:return} \EndIf
    \State Let $\OPT_{1B}$ (resp., $\OPT_{1C}$) be a minimum spanning sugraphs of $G_1$ of type $B$ (resp., $C$) w.r.t.\ $\{u, v\}$ if any such subgraph exists and if it belongs to some 2EC spanning subgraph of $G$.
    \If{$\OPT_{1C}$ is defined and $|\OPT_{1C}|\leq |\OPT_{1B}| - 1$}\label{line:typeCcase}
        \State Let $G_2'':= (V(G_2), E(G_2)\cup\{uv\})$.
        \State Let $S_1:=\OPT_{1C}, S_2:=\RED(G_2'')\setminus\{uv\}$.
        \State \textbf{return} $S:= S_1\cup S_2\cup F$, where $F\subseteq E(G), |F|\leq 1$ is a minimum set of edges such that $S$ is 2EC.\label{line:typeCcase:return}
    \Else\label{line:typeABcase}
        \State Let $G_2''':= (V(G_2)\cup\{w\}, E(G_2)\cup\{uw, vw\})$, where $w$ is a dummy node.
        \State Let $S_1:= \OPT_{1B}, S_2:=\RED(G_2''')\setminus\{uv, vw\}$.
        \State \textbf{return} $S := S_1\cup S_2$.\label{line:typeABcase:return}
    \EndIf
    \EndIf
    \State \textbf{return} $\ALG(G)$\label{line:baseCaseAlg}.
\end{algorithmic}
\label{alg:RED}
\end{algorithm}

Our reduction is described in Algorithm~\ref{alg:RED}. In the execution of Algorithm~\ref{alg:RED}, for every $W\subseteq V(G)$ and $S\subseteq G|W$, we interpret edges of $S$ as edges of $G$ after decontracting $W$. The proof of Lemma~\ref{lem:preprocessing} is implied by the following Lemmas~\ref{lem:redCorrectness},~\ref{lem:redRunningTime} and~\ref{lem:redCost}.

For the analysis of the algorithm, in particular when the condition at Line~\ref{line:hardCase} is met, it is convenient to define $\OPT_{1A}$ as the minimum spanning subgraph of $G_1$ of type $A$ w.r.t.\ $\{u, v\}$, if it exists (similarly as $\OPT_{1B}$ and $\OPT_{1C}$). We also define $\OPT_{2A}, \OPT_{2B}, \OPT_{2C}$ as the minimum spanning subgraph of $G_2$ of type $A, B$ or $C$, respectively, if it exists. We also let $\opt_{1A} := |\OPT_{1A}|$ if $\OPT_{1A}$ is defined, and $\opt_{1A} := \infty$ otherwise. We define analogously $\opt_{1B}, \opt_{1C}, \opt_{2A}, \opt_{2B}$ and $\opt_{2C}$. Since there are no edges between $V_1$ and $V_2$, $\OPT = (\OPT\cap G_1)\dot\cup (\OPT\cap G_2)$. Consequently, we define $\OPT_i$ as $\OPT\cap G_i$ and we let $\opt_i = |\OPT_i|$, for $i\in\{1, 2\}$. Notice that $\opt(G)= \opt_1 + \opt_2$, $\opt_i \geq |V(G_i)| - 1 \geq 3$, and, by Lemma~\ref{lem:redTypes}, $\OPT_i$ is of type $A, B$ or $C$, so $\opt_i \geq \min\{\opt_{iA}, \opt_{iB}, \opt_{iC}\}$, for $i\in\{1, 2\}$. In Line~\ref{line:partition} we use the following simple lemma.

\begin{lemma}\label{lem:redPartition}
    Let $\{u, v\}$ be a non-isolating $2$-vertex cut of a graph $G$. Then, if $|V(G)|\geq 6$, in polynomial time one can find a partition $(V_1, V_2)$ of $V\setminus\{u, v\}$ such that $2\leq |V_1|\leq |V_2|$ and there are no edges between $V_1$ and $V_2$.
\end{lemma}
\begin{proof}
    Since $\{u, v\}$ is a $2$-vertex cut, removing $\{u, v\}$ from $G$ splits the graph into $k$ connected components $C_1,\dots, C_k$ with $k\geq 2$ and $|V(C_1)|\leq |V(C_2)|\leq\dots\leq |V(C_k)|$. These components can be found in polynomial time. Moreover, since $\{u, v\}$ is non-isolating, if $k=2$ then $|V(C_1)|\geq 2$. If $k=2$, the partition $(V_1, V_2) =  (V(C_1), V(C_2))$ is of the desired type. If $k\geq 4$, the partition $(V_1, V_2) = (V(C_1)\cup V(C_2), \bigcup_{i=3}^k V(C_k))$ is of the desired type. Otherwise, i.e. $k = 3$, the partition $(V_1, V_2) = (V(C_1)\cup V(C_2), V(C_3))$ is of the desired type, because $|V(C_3)|\geq 2$, otherwise $|V(G)|\leq 5$, a contradiction to the assumption of the lemma.
\end{proof}

\begin{lemma}\label{lem:redCorrectness}
    Algorithm $\RED$ returns a 2EC spanning subgraph of $G$ if $\ALG$ does so. 
\end{lemma}
\begin{proof}
    We prove this by induction on $|E(G)|$. The base cases are given by Lines~\ref{line:baseCaseBruteForce} and~\ref{line:baseCaseAlg}. If Lines~\ref{line:parallel},~\ref{line:oneCut} and~\ref{line:irrelevant} are executed, the claim follows by induction and Lemmas~\ref{lem:parallel},~\ref{lem:oneCut} and~\ref{lem:irrelevantEdge}, resp. If Line~\ref{line:contract} is executed, the claim follows by induction, Facts~\ref{fact:contract} and~\ref{fact:uncontract}, and the fact that $H$ is 2EC by the definition of $\alpha$-contractible. From now on we focus on the case where the condition at Line~\ref{line:hardCase} is met. Notice that a partition as in Line~\ref{line:partition} exists by Lemma~\ref{lem:redPartition}, because since the condition at Line~\ref{line:baseCaseBruteForce} was not met, $|V(G)|\geq 6$. We will need the following claim:
                            
    \begin{claim}\label{claim:redNoTypeA}
        If $|V(G_1)|\leq 2/(\alpha - 1)$, there exists an optimum solution $\OPT$ such that $\OPT_1$ is of type $B$ or $C$. 
    \end{claim}
    \begin{proof}
        Assume $\OPT_1$ is of type $A$, otherwise the claim holds. If $\opt_{1A} \geq \opt_{1B} + 1$ then it must be that $\OPT_{2}$ is of type $C$, otherwise by Lemma~\ref{lem:connectRedTypes} $\OPT_{1B}\cup\OPT_2$ is a an optimum solution of cardinality less than $\OPT_1\cup\OPT_2$, a contradiction. Since by Lemma~\ref{lem:redTypes},  $\opt_{2B}\leq \opt_{2C} + 1$, one has that $|\OPT_{1B}\cup\OPT_{2B}| = \opt_{1B} + \opt_{2B}\leq \opt_{1A} - 1 + \opt_{2C} + 1 = |\OPT_{1A}\cup\OPT_{2C}|$, and thus by Lemma~\ref{lem:connectRedTypes} $\OPT_{1B}\cup\OPT_{2B}$ is an optimum solution such that $\OPT_1$ is of type $B$. 
        
        Assume now $\opt_{1A} \leq \opt_{1B}$. Since $|V(G_1)|\leq 2/(\alpha - 1)$, and Line~\ref{line:contract} was not executed, $\OPT_{1A}$ is not an $\alpha$-contractible subgraph of $G$. By Lemma~\ref{lem:redTypes}, this together with the fact that $\OPT_{1A}$ is a 2EC subgraph of $G$, implies that $\opt_{1A} > \alpha\cdot\min\{\opt_{1B}, \opt_{1C}\}$. Since we assumed $\opt_{1A} \leq \opt_{1B}$, one has that $\opt_{1A} > \alpha\cdot\opt_{1C}$, and by Lemma~\ref{lem:redTypes}, $\alpha \cdot\opt_{1C} < \opt_{1A} \leq \opt_{1B} \leq \opt_{1C} + 1$, which implies that $\opt_{1A} = \opt_{1B} = \opt_{1C} + 1$ and $\opt_{1C} < 1/(\alpha - 1) \leq 5$, because we have $\alpha\geq 6/5$. Since $\opt_{1C}\geq 3$, it must be that $\opt_{1C}\in \{3, 4\}$.
    
        Recall that $\OPT_{1C}$ consists of two 2EC components, one containing $u$ and one containing $v$. Those components either consist of a single node (namely, $u$ or $v$) or have at least $3$ nodes and consequently $3$ edges. Since $\opt_{1C}\leq 4$, it must be that $\OPT_{1C}$ consist of a component containing a single node (assume w.l.o.g.\ $v$) and a 2EC component containing $u$. Notice that by the minimality of $\OPT_{1C}$, the component containing $u$ must be a cycle $ux_1\dots x_k, k\in \{2, 3\}$. Since $G_1$ contains a subgraph of type $A$, there must be at least two edges from $v$ to $x_i, i\in [1, k]$, so at least one of them is incident to $x_1$ or $x_k$, say w.l.o.g.\ to $x_1$. Then the path $vx_1\dots x_ku$ implies $\opt_{1B} = \opt_{1C}$, a contradiction.
    \end{proof}

    Consider first the case when Line~\ref{line:bothBigCase:return} is executed. By Fact~\ref{fact:contract}, $G_i'=G_i|\{u, v\}$ is 2EC, for $i\in\{1, 2\}$, and it has strictly less edges than $G$, so $\RED(G_i')$ is 2EC by induction (when $\RED(G_i')$ is interpreted as a subgraph of $G_i')$. Notice that this implies that $\RED(G_i')$ is of type $A, B$ or $C$ w.r.t.\ $\{u, v\}$ (when $\RED(G_i')$ is interpreted as a subgraph of $G$ after decontracting $\{u, v\}$). Thus, $S_i$ is of type $A, B$ or $C$, for $i\in\{1, 2\}$. If $S_i$ is of type $C$, then by Lemma~\ref{lem:redTypes}, there is an edge $e_i$ such that $S_1\cup\{e_i\}$ is of type $B$, for $i\in\{1, 2\}$. Therefore, there is a set of edges $F\subseteq E(G)$, $|F|\leq 2$, such that $S_1\cup S_2\cup F = S_1'\cup S_2'$, where $S_i'$ is a spanning subgraph of $G_i$ of type $A$ or $B$, for $i\in \{1, 2\}$. Thus, $S_1\cup S_2\cup F$ is a 2EC spanning subgraph of $G$ by Lemma~\ref{lem:connectRedTypes}.

    Consider now the case when Line~\ref{line:typeCcase:return} is executed. 
    By Lemma~\ref{lem:redTypes}, $G_2$ has a subgraph of type $A$ or $B$ w.r.t.\ $\{u, v\}$ and hence $G_2$ itself is of type $A$ or $B$. This implies that $G_2'' = (V(G_2), E(G_2)\cup\{uv\})$ is 2EC, and since it has strictly less edges than $G$, by induction $\RED(G_2'')$ is 2EC. Notice that this implies $\RED(G_2'')\setminus\{uv\}$ is of type $A$ or $B$ w.r.t.\ $\{u, v\}$. By Lemma~\ref{lem:redTypes} there exists an edge $e$ such that $S_1\cup\{e\}$ is of type $B$. Therefore, there is a set of edges $F\subseteq E(G), |F|\leq 1$, such that $S_1\cup S_2\cup F = S_1'\cup S_2$, where $S_1'$ is a spanning subgraph of $G_1$ of type $A$ or $B$ w.r.t.\ $\{u, v\}$. Thus, $S_1\cup S_2\cup F$ is a 2EC spanning subgraph of $G$ by Lemma~\ref{lem:connectRedTypes}.

    Finally, assume that  Line~\ref{line:typeABcase:return} is executed. Since the condition at Line~\ref{line:bothBigCase} is not met, $|V(G_1)|\leq 2/(\alpha - 1)$. Thus we can assume w.l.o.g.\ $\OPT_1$ is of type $B$ or $C$ by Claim~\ref{claim:redNoTypeA}. By Lemma~\ref{lem:redTypes}, if $\OPT_{1C}$ is defined, $\OPT_{1B}$ must also be defined, so since $\OPT_1$ is of type $B$ or $C$, $S_1 = \OPT_{1B}$ is defined. Moreover, since $\OPT_1$ is of type $B$ or $C$, $\OPT_2$ is not of type $C$ by Lemma~\ref{lem:redTypes}. This implies that $G_2$ contains a subgraph of type $A$ or $B$ w.r.t.\ $\{u, v\}$. This implies $G_2''' = (V(G_2)\cup\{w\}, E(G_2)\cup\{uv, vw\})$ is 2EC, and since it has strictly less edges than $G$, by induction $\RED(G_2''')$ is 2EC. Notice that this implies $\RED(G_2''')\setminus\{uv, vw\}$ is of type $A$ or $B$ w.r.t.\ $\{u, v\}$. Since $S_1$ is of type $B$ and $S_2$ is of type $A$ or $B$, $S_1\cup S_2$ is a 2EC spanning subgraph of $G$ by Lemma~\ref{lem:connectRedTypes}. 
\end{proof}

\begin{lemma}\label{lem:redRunningTime}
    Algorithm $\RED$ runs in polynomial time in $|V(G)|$ if $\ALG$ does so. 
\end{lemma}
\begin{proof}
    Let $n = |V(G)|$. Let us first show that each recursive call, excluding the corresponding subcalls, can be executed in polynomial time. $1$-vertex cuts, parallel edges or loops can easily be found in polynomial time. Irrelevant edges can also be found in polynomial time by enumerating all $2$-vertex cuts. To find $\alpha$-contractible graphs of size at most $2/(\alpha - 1)$, we start by enumerating all subgraphs of $G$ of size at most $2/(\alpha - 1)$, which can be done in polynomial time since $\alpha\geq 6/5$. Now, such a subgraph $H$ is $\alpha$-contractible if and only if $H$ is 2EC and for all sets $H'\subseteq G[V(H)]$ such that $|H'| < |H|/\alpha$, $(G\setminus G[V(H)])\cup H'$ is not 2EC. Since $\alpha\geq 6/5$, $|H|\leq 2/(\alpha - 1)\leq 10$, so every such subset can be found in constant time by exhaustive enumeration, and we can check if a graph is 2EC in polynomial time. 
        
    The partition $(V_1, V_2)$ in Line~\ref{line:partition} can be found in polynomial time via Lemma~\ref{lem:redPartition}, because since the condition at Line~\ref{line:baseCaseBruteForce} was not met, $|V(G)|\geq 6$. The sets $F$ at Line~\ref{line:bothBigCase:return} and Line~\ref{line:typeCcase:return} can be found in polynomial time by trying all sets of edges of size at most $2$ and $1$, respectively, and checking if the resulting graph is 2EC, which can be done in polynomial time.

    It is left to bound the recursive calls when the condition at Line~\ref{line:hardCase} is met. Let $n_1 = |V(G_1)|$ and $n_2 = |V(G_2)|$. Notice that $n_1 + n_2 = n + 2$. Since $|V_1| \leq |V_2|$, $n_1\leq (n - 2)/2 + 2 = n/2 + 1$, and since $|V_1|\geq 2$, $n_1\geq 4$ and thus $n_2\leq n-2$. Also, notice that $|V(G_i')| = n_i - 1$, for $i\in\{1, 2\}$, $|V(G_2'')| = n_2$ and $|V(G_2''')| = n_2 + 1$. Since $\alpha\geq 6/5$, $\OPT_{1B}$ and $\OPT_{1C}$ can be computed in constant time when $|V(G_1)|\leq 2/(\alpha - 1)$. Let $f(n)$ be the running time of our algorithm. If Line~\ref{line:bothBigCase}, is executed, $f(n)\leq f(n_1 - 1) + f(n_2 - 1) + \poly(n)\leq \max_{k\in [4, n/2 + 1]}\{f(k - 1) + f(n - k + 1)\} + \poly(n)$. If Line~\ref{line:typeCcase:return} is executed $f(n)\leq f(n_2) +\poly(n)\leq f(n -2) + \poly(n)$.  If Line~\ref{line:typeABcase:return} is executed, $f(n)\leq f(n_2 + 1) + \poly(n)\leq f(n - 1) + \poly(n)$. In all cases $f(n)$ is polynomially bounded.
\end{proof}

\begin{lemma}\label{lem:redCost}
$|\RED(G)|\leq\begin{cases}
    \opt(G), \quad & \text{if } |V(G)|< \frac{4}{\alpha-1};\\
    \alpha\cdot \opt(G) - 2, \quad & \text{if } |V(G)|\geq \frac{4}{\alpha-1}.
\end{cases}$
\end{lemma}
\begin{proof}
We prove the claim by induction on $|E(G)|$. The base cases correspond to Lines~\ref{line:baseCaseBruteForce} and~\ref{line:baseCaseAlg}. Notice, at Line~\ref{line:baseCaseBruteForce}, that if $4/(\alpha-1) < 5$ and $|V(G)|\geq 4/(\alpha - 1)$, then we return $\opt(G)$ when we are required to return $\alpha\cdot\opt(G) - 2$. However, since $\opt(G)\geq |V(G)|\geq 4/(\alpha - 1) > 2/(\alpha - 1)$, $\opt(G) < \alpha\cdot\opt(G) - 2$. The claim follows by induction and by Lemmas~\ref{lem:parallel},~\ref{lem:oneCut} and~\ref{lem:irrelevantEdge} when executing Lines~\ref{line:parallel},~\ref{line:oneCut} and ~\ref{line:irrelevant}, respectively.

Consider the execution at Line~\ref{line:contract}. Since $H$ is an $\alpha$-contractible graph, $\opt(G)$ contains at least $|H|/\alpha$ edges of $H$, and by Fact~\ref{fact:contract}, $\OPT(G)|H$ is a 2EC spanning subgraph of $G|H$. Therefore, $\opt(G|H)\leq \opt(G) - |H|/\alpha$. Now, if $V(G|H) \geq 4/(\alpha - 1)$, by the induction hypothesis we return a solution of cardinality at most $|H| + \alpha\cdot\opt(G|H) - 2\leq \alpha\cdot\opt(G) - 2$. If $V(G|H) < 4/(\alpha - 1)$, then by induction hypothesis we return a solution of cost $|H| + \opt(G|H)\leq \alpha\cdot\opt(G) - (\alpha - 1)\cdot\opt(G|H) < \alpha\cdot\opt(G) - 2$. The last inequality comes from the fact that, since we did not execute Line~\ref{line:baseCaseBruteForce}, then $|V(G)| > 4/(\alpha - 1)$. Since one also has $|V(H)| \leq 2/(\alpha -1)$, one has $\opt(G|H)\geq |V(G|H)| > 4/(\alpha - 1) - 2/(\alpha -1) = 2/(\alpha -1)$. 

We now focus on the case where the condition at Line~\ref{line:hardCase} holds. If $|V(G_2)|\leq 4/(\alpha - 1)$ then Line~\ref{line:bruteForce2} is executed and we return $\opt(G)$. Since Line~\ref{line:baseCaseBruteForce} was not executed, $\opt(G)\geq |V(G)| > 4/(\alpha - 1)$, and thus $\opt(G) < \alpha\cdot \opt(G) - 2$. From now on we assume $|V(G_2)| > 4/(\alpha - 1)$. We will need the following claim:

\begin{claim}\label{claim:redBounds}
    The following statements are true.
    \begin{enumerate}
        \item If $G_i'$ is 2EC, $\opt(G_i') \leq \min\{\opt_{iA}, \opt_{iB}, \opt_{iC}\}$, for $i\in\{1, 2\}$,
        \item If $G_2''$ is 2EC, $\opt(G_2'') \leq \min\{\opt_{2A}, \opt_{2B} + 1\}$,
        \item If $G_2'''$ is 2EC, $\opt(G_2''') \leq \min\{\opt_{2A}, \opt_{2B}\} + 2$.
    \end{enumerate}
\end{claim}
\begin{proof}
    Notice that $\OPT_{iA}|\{u, v\}, \OPT_{iB}|\{u, v\}$ and $\OPT_{iC}|\{u, v\}$ are all 2EC spanning subgraphs of $G_i'$ (if $\OPT_{iA}, \OPT_{iB}$ and $\OPT_{iC}$ are defined, respectively), so $\opt(G_i') \leq \min\{\opt_{iA}, \opt_{iB}, \opt_{iC}\}$. Similarly, $\OPT_{2A}$ and $\OPT_{2B}\cup\{uv\}$ are 2EC spanning subgraphs of $G_2'' = (V(G_2), E(G_2)\cup\{uv\})$ (if $\OPT_{2A}$ and $\OPT_{2B}$ are defined, respectively), so $\opt(G_2'') \leq \min\{\opt_{2A}, \opt_{2B} + 1\}$. Finally, $\OPT_{2A}\cup\{uw, vw\}$ and $\OPT_{2B}\cup\{uw, vw\}$ are 2EC spanning subgraphs of $G_2''' = (V(G_2)\cup\{w\}, E(G_2)\cup\{uw, vw\})$ (if $\OPT_{2A}$ and $\OPT_{2B}$ are defined, respectively), so $\opt(G_2''') \leq \min\{\opt_{2A}, \opt_{2B}\} + 2$.
\end{proof}

\begin{caseanalysis}
    \case{The condition at Line~\ref{line:bothBigCase} is met.} By Fact~\ref{fact:contract}, $G_i'=G_i|\{u, v\}$ is 2EC, and it has strictly less edges than $G$, so we can apply the induction hypothesis when executing $\RED$ on $G_i$, for $i\in\{1, 2\}$. 
    
    One has that $|V(G_2)| > 4/(\alpha -1)$, so $|V(G_2')| = |V(G_2)| - 1 > 4/(\alpha - 1) - 1$, implying $|V(G_2')|\geq 4/(\alpha - 1)$. Thus, by Claim~\ref{claim:redBounds}, and induction hypothesis, $|S_2| = |\RED(G_2')|\leq \alpha\cdot\min\{\opt_{2A}, \opt_{2B}, \opt_{2C}\} - 2\leq \alpha\cdot\opt_2 - 2$. The last inequality follows by Lemma~\ref{lem:redTypes}. Since the condition at Line~\ref{line:bothBigCase} is met, $\opt(G_1')\geq |V(G_1')| = |V(G_1)| - 1 > 2/(\alpha - 1) - 1$, implying $\opt(G_1')\geq 2/(\alpha - 1)$. Thus, $\opt(G_1')\leq \alpha\cdot \opt(G_1') - 2$. Hence, by Claim~\ref{claim:redBounds} and induction hypothesis, we also have $|S_1|= |\RED(G_1')|\leq \alpha\cdot\opt(G_1') - 2\leq \alpha\cdot\min\{\opt_{1A}, \opt_{1B}, \opt_{1C}\} - 2\leq \alpha\cdot\opt_1 - 2$.  The last inequality follows by Lemma~\ref{lem:redTypes}. One has $|S_1| + |S_2| + |F|\leq  \alpha\cdot\opt_1 - 2 +  \alpha\cdot \opt_2 - 2 + 2 = \alpha\cdot\opt -2$.
    
    \case{The condition at Line~\ref{line:typeABcase} is met.} Since the condition at Line~\ref{line:bothBigCase} is not met, $|V(G_1)|\leq 2/(\alpha - 1)$. Thus, $\OPT_1$ is w.l.o.g.\ of type $B$ or $C$ by Claim~\ref{claim:redNoTypeA}. Since $\OPT_1$ is of type $B$ or $C$, $\OPT_2$ is not of type $C$ by Lemma~\ref{lem:redTypes}. This implies that $G_2$ contains a subgraph of type $A$ or $B$ w.r.t.\ $\{u, v\}$. Hence $G_2''' = (V(G_2)\cup\{w\}, E(G_2)\cup\{uv, vw\})$ is 2EC, and since it has strictly less edges than $G$, we can apply the induction hypothesis when executing $\RED$ on $G_2'''$.
    
    Notice that every 2EC spanning subgraph of $G_2'''$ must include the edges $uw, vw$, because $w$ has degree $2$ in $G_2'''$. By induction hypothesis and since $|V(G_2)| > 4/(\alpha - 1)$, $|S_2| = |\RED(G_2''')\setminus\{uw, vw\}| = |\RED(G_2''')| - 2 \leq \alpha\cdot\opt(G_2''') - 4$. Since Line~\ref{line:bothBigCase} was not executed, by Claim~\ref{claim:redNoTypeA} we can assume $\OPT_1$ is of type $B$ or $C$. Since the condition at Line~\ref{line:typeCcase} is not met, one has $|S_1|\leq\opt_1$. One has $|S_1| + |S_2| \leq \opt_1 + \alpha\cdot\opt(G_2''') - 4 \leq \opt_1 + \alpha\cdot(\opt_2 + 2) - 4 = \alpha\cdot\opt - 2 + (2-\opt_1)\cdot(\alpha - 1)\leq \alpha\cdot\opt - 2$. The second inequality follows from Claim~\ref{claim:redBounds}. The last inequality follows from the fact that $\opt_1\geq 3$ and $\alpha\geq 6/5 >1$.
    
    \case{The condition at Line~\ref{line:typeCcase} is met.} Since $\OPT_{1C}$ is defined, it belongs to a 2EC spanning subgraph of $G$. Therefore, by Lemma~\ref{lem:redTypes}, there is a subgraph of $G_2$ of type $A$. This implies that $G_2'' = (V(G_2), E(G_2)\cup\{uv\})$ is 2EC, and since it has strictly less edges than $G$, we can apply the induction hypothesis when executing $\RED$ on $G_2''$.

    We claim that $|S_2| + |F| \leq |\RED(G_2'')|$. By Lemma~\ref{lem:redCorrectness}, $\RED(G_2'')$ is 2EC. Notice that this implies $\RED(G_2'')\setminus\{uv\}$ is of type $A$ or $B$ w.r.t.\ $\{u, v\}$. If $\RED(G_2'')\setminus\{uv\}$ is of type $A$ w.r.t.\ $\{u, v\}$, then $S_1\cup S_2$ is 2EC and thus $F=\emptyset$ satisfies that $S_1\cup S_2\cup F$ is 2EC, implying $|S_2| + |F|  = |\RED(G_2'')\setminus\{uv\}| \leq |\RED(G_2'')|$. If $\RED(G_2'')\setminus\{uv\}$ is of type $B$ w.r.t.\ $\{u, v\}$, then since $\RED(G_2'')$ is 2EC this implies that $uv\in \RED(G_2'')$. Therefore, $|S_2| + |F|  \leq |\RED(G_2'')\setminus\{uv\}| + 1 = |\RED(G_2'')| - 1 + 1 = |\RED(G_2'')|$. Thus, $|S_2| + |F| \leq |\RED(G_2'')|$, as we wanted to prove. By induction hypothesis, $|S_2| + |F| \leq |\RED(G_2'')| \leq \alpha\cdot\opt(G_2'') - 2$. Also, since Line~\ref{line:bothBigCase} was not executed, by Claim~\ref{claim:redNoTypeA} we can assume $\OPT_1$ is of type $B$ or $C$.

    Now, if $\OPT_1$ is of type $C$, then  by Lemma~\ref{lem:redTypes}, $\OPT_2$ is of type $A$. Notice that $|S_1| = \opt_1$. One has $|S_1| + |S_2| +|F|\leq \opt_1 + \alpha\cdot\opt(G_2'') - 2\leq \opt_1 + \alpha\cdot\opt_2 - 2 = \alpha\cdot\opt(G) - 2 - (\alpha - 1)\cdot\opt_1 < \alpha\cdot\opt(G) - 2$, where the second inequality follows from Claim~\ref{claim:redBounds} and the fact that $\OPT_2$ is of type $A$, and the last inequality from the fact that $\alpha\geq 6/5>1$.
    
    On the other hand, if $\OPT_1$ is of type $B$, then by Lemma~\ref{lem:redTypes}, $\OPT_2$ is not of type $C$. By the condition at Line~\ref{line:typeCcase}, $|S_1|\leq \opt_1 - 1$. One has $|S_1| + |S_2| +|F|\leq \opt_1 - 1 + \alpha\cdot\opt(G_2'') - 2\leq \opt_1 - 1 + \alpha\cdot(\opt_2 + 1) - 2 = \alpha\cdot\opt(G) - 2 + (\alpha - 1)\cdot(1 - \opt_1) < \alpha\cdot\opt(G) - 2$, where the second inequality follows from Claim~\ref{claim:redBounds} and the last inequality from the fact that $\opt_1\geq 3$ and $\alpha\geq 6/5>1$.
    \end{caseanalysis}
\end{proof}

\section{Computation of a Canonical 2-Edge Cover}
\label{sec:lowerBound}

In this section, we discuss the computation of a canonical $2$-edge cover. The starting point of our construction is the computation of a minimum triangle-free $2$-edge cover $H$. We recall that a \emph{$2$-edge cover} $H$ is a subset of edges such that each node $v$ has at least $2$ edges of $H$ incident to it. Such $H$ is \emph{triangle-free} if no 2EC component of $H$ is a triangle (notice that triangles might still be induced subgraphs of $H$).  

Kobayashi and Noguchi~\cite{KN23} observed that the computation of a minimum triangle-free 2-edge cover can be reduced to the computation of a maximum triangle-free $2$-matching. Recall that a \emph{2-matching} $M$ is a subset of edges such that at most $2$ edges of $M$ are incident to each node. $M$ is \emph{triangle-free} if it does not contain any triangle as an induced subgraph. In particular, the following two lemmas hold:
\begin{lemma}[\cite{KN23}]\label{lem:fromMatchingToCover}
    Let $G = (V, E)$ be a connected graph such that the minimum degree is at least $2$ and $|V|\geq 4$. Given a triangle-free $2$-matching $M$ of $G$, one can compute in polynomial time a triangle-free $2$-edge cover $H$ of $G$ such that $|H|\leq 2|V| - |M|$.
\end{lemma}
\begin{lemma}[\cite{KN23}]\label{lem:fromCoverToMatching}
    Given a triangle-free $2$-edge cover $H$ of a graph $G = (V, E)$, one can compute in polynomial time a triangle-free $2$-matching $M$ of $G$ such that $|M|\geq 2|V| - |H|$.
\end{lemma}
A corollary of the above lemmas is that, given a minimum triangle-free $2$-edge cover $H^*$ and a maximum triangle-free $2$-matching $M^*$, then $|H^*|=2|V|-|M^*|$. Another consequence is that, given a maximum triangle-free $2$-matching $M^*$, one can compute a minimum triangle-free $2$-edge cover $H^*$. Indeed, one can also show that a PTAS for the first problem implies a PTAS for the latter one: 
\begin{restatable}{lemma}{lemcomputeTriFree}\label{lem:computeTriFree2EdgeCover}
 Let $M$ be a triangle-free $2$-matching of a 2EC graph $G$, such that $|M| \geq (1 - \eps)|M^*|$, where $M^*$ is a maximum triangle free $2$-matching of $G$ and $1\geq \eps\geq 0$. Given $M$, one can compute in polynomial time a triangle-free $2$-edge cover $H$ of $G$ such that $|H|\leq (1 + \eps)|H^*|$, where $H^*$ is a minimum triangle-free $2$-edge cover of $G$.
\end{restatable}
\begin{proof}
    First let us show that $|M^*| + |H^*| = 2|V|$. Apply Lemma~\ref{lem:fromMatchingToCover} to find a triangle-free $2$-edge cover $H'$ such that $|H'|\leq 2|V| - |M^*|$. Thus, $|M^*| + |H^*|\leq |M^*| + |H'| \leq 2|V|$. Similarly, apply Lemma~\ref{lem:fromCoverToMatching} to find a triangle-free $2$-matching $M'$ such that $|M'|\geq 2|V| - |H^*|$, so that $|M^*| + |H^*|\geq |M'| + |H^*| \geq 2|V|$.

    The claim of the lemma now follows by applying Lemma~\ref{lem:fromMatchingToCover} to $M$ to efficiently find a triangle-free $2$-edge cover $H$ such that $|H|\leq 2|V| - |M|\leq 2|V| - (1-\eps) |M^*| = |M^*| + |H^*| - (1-\eps) |M^*| = |H^*| + \eps |M^*|\leq (1+\eps)|H^*|$. The last inequality comes from the fact that every $2$-edge cover of $G$ has cardinality at least $|V|$, and every $2$-matching of $G$ has cardinality at most $|V|$, so $|M^*|\leq |H^*|$.
\end{proof}

Until very recently, the only known polynomial-time algorithm for maximum triangle-free $2$-matching was contained in the Ph.D. thesis of Hartvigsen from 1984~\cite{Hartvigsen1984extensions}. This result is extremely complicated, which explains why it was not used as a subroutine to approximate 2ECSS in most prior work. Recently, the situation has changed. First of all, a journal version (still very complex) of the previous result appeared in \cite{H24}. Second, Paluch \cite{K23} published a technical report with an alternative proof of the same result (again, far from trivial). Third, Bosch-Calvo, Grandoni, and Jabal Ameli~\cite{BGJ23} published a technical report 
containing a PTAS for the same problem, which is much simpler than the previous two results, yet non-trivial. Kobayashi and Noguchi subsequently presented a simpler analysis of the aforementioned PTAS~\cite{KN25}.

Notice that, using the latter PTAS rather than the exact algorithms in \cite{H24,K23}, with our approach one would get a $5/4+\eps$ approximation for 2ECSS (instead of a $5/4$-approximation). 

We already discussed how to impose that the initial triangle-free $2$-edge cover $H$ contains at least one component with at least $8$ nodes, while still guaranteeing that $|H|\leq \opt$. 
We now give the proofs for Lemmas~\ref{lem:canonical2EdgeCover} and~\ref{lem:costCanonical}. The first one shows that we can assume that the initial triangle-free $2$-edge cover is canonical, the second one that the initial credit is enough to satisfy our credit invariant.

In the next lemma (and in later steps) we will sometimes \emph{merge} two or more connected components of the current triangle-free $2$-edge cover, meaning that the respective node sets belong to the same connected component after the merge. At the same time, we never disconnect at a later steps nodes which are connected at a certain step. This way, the existence of a component with at least $8$ nodes at the beginning of the process guarantees that such a component also exists at each following step. 
More formally, given two subset of edges $H_1$ and $H_2$, we say that $H_1$ is a \emph{component coarsening} of $H_2$ if the following happens: given the node set $W_2$ of any connected component induced by $H_2$, then $W_2$ is a subset of the node set $W_1$ of some connected component induced by $H_1$. 

\lemCanonical*
\begin{proof}   
    We start with $H':=H$ and do the following process: If there are sets of edges $F, F', |F|, |F'|\leq 5$ such that $H'':=(H'\setminus F)\cup F'$ is a triangle-free $2$-edge cover satisfying:
    \begin{enumerate*}[label={(\arabic*)}]
        \item $|H''|\leq |H'|$,
        \item $H''$ is a component coarsening of  $H'$,
        \item Either $|H''|< |H'|$, or $H''$ has strictly less components than $H'$, or the number of components in $H''$ and $H'$ is the same and the total number of bridges of $H''$ is strictly less than the number of bridges of $H'$.
    \end{enumerate*}
    Then update $H':=H''$. Clearly, we can find in polynomial time such sets $F, F'$, if they exist, and this process ends in a polynomial number of steps. Obviously, by construction, the final $H'$ satisfies $|H'|\leq |H|$. We next show that $H'$ is canonical. 
    Obviously $H'$ is triangle-free, hence property (1) holds. $H'$ is a component coarsening of $H$, hence the presence of a component with at least $8$ nodes in $H$ guarantees that such a component also exists in $H'$ (thus property (5) holds). The next claim implies properties (3) and (4). 
    \begin{claim}\label{clm:blocksCanonical}
        Let $B$ be a 2EC block of a complex component $C$ of $H'$ such that there is exactly one node of $B$ adjacent in $C$ to nodes in $V(C)\setminus V(B)$. Then $|B|\geq 6$.
    \end{claim}
    \begin{proof}
        Assume this is not the case, so $|B| \leq 5$. It must be that $B$ is a cycle, otherwise, by a simple case analysis, $B$ is a $4$-cycle with a chord $e$, and $H'\setminus\{e\}$ is a triangle-free $2$-edge cover that is a component coarsening of $H$, a contradiction. Let $B = u_1u_2\dots u_k, k\in\{3, 4, 5\}$, where $u_1$ is the node of $B$ adjacent to nodes in $V(C)\setminus V(B)$. By the Hamiltonian pairs Lemma~\ref{lem:hamiltonianPairs}, there are two distinct Hamiltonian pairs of nodes $\{v_1, w_1\}, \{v_2, w_2\}$ of $B$: recall that by definition there is a Hamiltonian path between $v_j$ and $w_j$ in $G[V(B)]$ and both nodes are adjacent to nodes in $V(G)\setminus V(B)$. Notice that we can assume that $u_1$ is a node of one such pair. Indeed otherwise, since $|\{v_1, w_1, v_2, w_2\}|\geq 3$ and $B$ is a $k$-cycle with $k\leq 5$, at least one node $x\in \{v_1, w_1, v_2, w_2\}$ is adjacent to $u_1$ in $B$. Hence $\{x,u_1\}$ is a valid Hamiltonian pair. Assume w.l.o.g.\ that there is a Hamiltonian $u_1, v_1$-path $P$ in $G[V(B)]$, and $v_1$ is adjacent to a node $w\in V(G)\setminus V(B)$. Let $H'':= (H'\setminus B)\cup P\cup\{v_1w\}$, and notice that $|H''|= |H'|$. If $w$ is not in $C$, then $H''$ is a triangle-free $2$-edge cover that is a component coarsening of $H$ with strictly less components than $H'$, a contradiction. Otherwise, i.e. $w\in V(C)\setminus V(B)$, $H''$ is a component coarsening of $H'$ with the same number of components as $H'$. However, in $H''$ all the bridges and the 2EC blocks of $H'$ that shared an edge with any path from $u_1$ to $w$ in $C\setminus B$ become part of the same 2EC block together with $P\cup\{v_1w\}$, and all the other bridges and 2EC blocks remain the same. 
        Since $u_1$ is adjacent to nodes in $V(C)\setminus V(B)$, at least one bridge belongs to that path. This is a contradiction as the total number of bridges of $H''$ is less than in $H'$.
    \end{proof}
    \begin{corollary}
    Given a complex component $C$ of $H'$, every 2EC block of $C$ contains at least $4$ edges, and that there are at least two 2EC blocks of $C$ with at least $6$ edges.    
    \end{corollary}
    \begin{proof}
    Consider the first statement. Notice that if a 2EC block $B'$ of a complex component $C$ has $|B'| < 4$, then $B'$ is a triangle $u_1u_2u_3$. Claim \ref{clm:blocksCanonical} states that, if there is exactly one node of $B'$ adjacent to nodes in $V(C)\setminus V(B')$, then $|B'|\geq 6$. Since $|B| = 3$, this implies that there are at least two nodes of $B$, say $u_1, u_2$, adjacent to nodes in $V(C)\setminus V(B')$. But then $H'\setminus\{u_1u_2\}$ is a triangle-free $2$-edge cover that is a component coarsening of $H'$, a contradiction. 
    
    Consider the second statement. Observe that $C$ must contain at least two 2EC blocks with exactly one bridge incident to them. 
    Indeed, the graph resulting from the contraction of every 2EC block of $C$ into a single node must be a tree, and since $H'$ is a $2$-edge cover, each leaf of that tree must correspond to exactly one 2EC block.
    Notice that there are at least $2$ such leaves, and by Claim \ref{clm:blocksCanonical} they contain at least $6$ edges each.    
    \end{proof}    
    It is left to prove property (2), namely that every component of $H'$ with less than $8$ edges is a cycle. Assume to get a contradiction that $C$ is a component with $|C|\leq 7$ that is not a cycle. Since we proved that every complex component of $H'$ must contain at least two 2EC blocks with at least $6$ edges each, it must be that $C$ is 2EC, and thus $C$ must contain a cycle. Let $P_0$ be the cycle of $C$ with the minimum possible number of edges. 
    There is no edge $e\in E(C)$ that is a chord of $P_0$. Indeed, otherwise $H'\setminus\{e\}$ is a triangle-free $2$-edge cover that is a component coarsening of $H'$, a contradiction. This implies that $C$ contains a node $v$ (of degree at least $2$) not in $P_0$, hence $|P_0|\leq 5$. Let $P_0 = u_1u_2\dots u_k, k\in\{3, 4, 5\}$ and $P_1:= E(C)\setminus P_0$. Notice that $|P_1| \in\{2, 3, 4\}$.

We claim that $P_1$ is either a cycle sharing a node with $P_0$ or a path  between distinct nodes of $P_0$. Since $C$ is 2EC, by Fact~\ref{fact:contract}, $C|P_0$ is also 2EC, so in $C|P_0$ there is a cycle containing the node corresponding to the contracted $P_0$. This implies that $P_1$ must contain a subgraph $P$ that is either a cycle sharing a node with $P_0$ or a path of length at least two  between distinct nodes of $P_0$. We also note that $P$ is not a path incident to adjacent nodes in $P_0$. Indeed, otherwise $C$ contains a cycle with a chord, a contradiction as already argued before. Therefore $|P_0|+|P|\ge 6$, because if $P_0$ is a $3$-cycle, then $P$ must be a cycle and thus $|P|\ge 3$, while if $P_0$ is not a $3$-cycle then $|P_0|\geq 4$ and $|P|\geq 2$. Assume to get a contradiction that $P_1$ contains additional edges. Since $|P_0|+|P|\ge 6$ and $|C|\leq 7$, it must be that $|P_0| + |P| = 6$ and there is an edge $e\in C$ with both endpoints in $V(P_0)\cup V(P)$. 
But then, notice that $P_0\cup P$ is 2EC and therefore $C\setminus\{e\}$ is 2EC. This implies  that $H'\setminus\{e\}$ is a triangle-free $2$-edge cover that is a component coarsening of $H'$, a contradiction. 
    We now distinguish a few cases:

    \begin{caseanalysis}

    \case{$P_0$ is a $3$-cycle.} Since $P_1$ does not have as endpoints adjacent nodes in $P_0$, then $P_1$ is a cycle $v_1v_2\dots v_{l}, l\in\{3, 4\}$ with $v_1 = u_1$. 
    By the $3$-matching Lemma~\ref{lem:matchingOfSize3} applied to $(V(P_0), V(G)\setminus V(P_0))$, both $u_1$ and $u_2$ have edges to nodes in $V(G)\setminus V(P_0)$ in $G$. 
    Consider the edges $u_2 w$ and $u_3 x$, where $w, x \in V(G) \setminus V(P_0)$, $w \neq x$. 
    If $w\notin V(C)$, then $(H'\setminus\{u_1u_2\})\cup \{u_2w\}$ is a triangle-free $2$-edge cover that is a component coarsening of $H'$ with strictly less components than $H'$, a contradiction.
    Similarly, if $x \notin V(C)$, this also holds for $(H'\setminus\{u_1u_3\})\cup \{u_3x\}$. 
    Thus, $w, x \in \{v_2, v_3, v_4\}$. Therefore, since $w \neq x$, we have $|\{w, x \} \cap \{v_2, v_4 \}| \geq 1$. W.l.o.g. assume $w = v_2$.
    But then, $(H'\setminus\{v_1w, v_1u_2\})\cup \{u_2w\}$ is a triangle-free $2$-edge cover that is a component coarsening of $H'$ with fewer edges than $H'$, a contradiction. 

    \case{$P_0$ is a $4$-cycle.} Notice that $P_1$ cannot be a cycle, because since $|C|\leq 7$, it would be a $3$-cycle, a contradiction to the fact that $P_0$ has the minimum number of edges. Since the endpoints of $P_1$ are not adjacent in $P_0$ and $|C|\leq 7$, $P_1$ is a path of length $2$ or $3$ between non-adjacent nodes of $P_0$, say between $u_1$ and $u_3$, so $P_1= v_1v_2\dots v_l$ with $v_1 = u_1, v_l = u_3, l\in\{3, 4\}$. By the $3$-matching Lemma~\ref{lem:matchingOfSize3} applied to $(V(P_0), V(G)\setminus V(P_0))$, there is an edge $uw$ from $u\in \{u_2, u_4\}$ to $w\in V(G)\setminus V(P_0)$. By symmetry, we can assume w.l.o.g.\ that $u = u_2$. If $w\notin V(C)$, then $(H'\setminus\{u_1u_2\})\cup\{u_2w\}$ is a triangle-free $2$-edge cover that is a component coarsening of $H'$ with strictly less components than $H'$, a contradiction. 
    Thus, it must be that $w\in\{v_2, v_3\}$, and if $w=v_3$ then $l = 4$. 
    Assume $w=v_2$, the other case being symmetric when $l=4$. Then $(H'\setminus\{u_1v_2, u_2u_3\})\cup\{u_2v_2\}$ is a triangle-free $2$-edge cover that is a component coarsening of $H'$ with fewer edges than $H'$, a contradiction.

    \case{$P_0$ is a $5$-cycle.} Then, since $|P_1|\geq 2$ and $|C|\leq 7$, $P_1$ is a path of length $2$. Also, $P_1$ is incident to non-adjacent nodes of $P_0$, so we can assume w.l.o.g.\ $P_1 = u_1vu_3, v\notin V(P_0)$. But then, $C$ contains a cycle $u_1u_2u_3v$ of length $4$, a contradiction to the fact that $P_0$ is the smallest cycle in $C$.
    \end{caseanalysis}
\end{proof}

\lemCostCanonical*
\begin{proof}
    We start by assigning $1/4$ credits to every edge of $H$, hence $|H|/4$ credits in total. Then we show how to distribute them so that each bridge, component and 2EC block receives at least the number of credits required by the credit assignment scheme. Then obviously $\cost(H)\leq 5/4\cdot |H|$. The $1/4$ credits of each bridge $b$ are assigned to $b$. Every 2EC component $C$ is assigned the credits of its edges, hence $|E(C)|/4$ credits for small components and at least $8/4=2$ credits for large ones. Since $H$ is canonical, every 2EC block $B$ of $H$ contains at least $4$ edges: each such $B$ keeps $1$ credit among the $|E(B)|/4\geq 1$ credits of its edges. The remaining credits of the edges of $B$ are assigned to the complex component $C$ containing $B$. Since $H$ is canonical, every complex component $C$ contains at least $2$ blocks $B_1, B_2$ with at least $6$ edges each. Therefore, those blocks assign at least $2 \cdot 2/4=4/4=1$ credits to $C$ as required. 
    \end{proof}

\section{Bridge Covering}\label{sec:bridgeCovering}

In this section we prove Lemma \ref{lem:bridgeCovering}, which we restate next. The proofs of this section are almost identical to those in~\cite{GGJ23soda}. We present them here for completeness, although we remark that we have only changed their proofs slightly to adapt to our notation.

\lemBridgeCovering*

Let $C$ be any complex component of $S$: notice that $C$ contains at least one bridge. Let $G_C$ be the multi-graph obtained from $G$ by contracting into a single node each block $B$ of $C$ and each connected component $C'$ of $S$ other than $C$. Let $T_C$ be the tree in $G_C$ induced by the bridges of $C$: we call the nodes of $T_C$ corresponding to blocks \emph{block nodes}, and the remaining nodes of $T_C$ \emph{lonely nodes}. Observe that the leaves of $T_C$ are necessarily block nodes (otherwise $S$ would not be a 2-edge cover). 

At high level, in this stage of our construction we will transform $S$ into a new solution $S'$ containing a component $C'$ spanning the nodes of $C$ (and possibly the nodes of some other components of $S$). More precisely, $S'$ will be a component coarsening of $S$: in particular, this will automatically guarantee the property that $S'$ contains a component with at least $8$ nodes. Furthermore, no new bridge is created and at least one bridge $e$ of $C$ is not a bridge of $C'$ (intuitively, the bridge $e$ is covered). 

A \emph{bridge-covering path} $P_C$ is any path in $G_C\setminus E(T_C)$ with its (distinct) endpoints $u$ and $v$ in $T_C$, and the remaining (internal) nodes outside $T_C$. Notice that $P_C$ might consist of a single edge, possibly parallel to some edge in $E(T_C)$. Augmenting $S$ along $P_C$ means adding the edges of $P_C$ to $S$, hence obtaining a new 2-edge cover $S'$. Notice that $S'$ obviously has fewer bridges than $S$: in particular all the bridges of $S$ along the $u, v$-path in $T_C$ are not bridges in $S'$ (we also informally say that such bridges are \emph{removed}), and the bridges of $S'$ are a subset of the bridges of $S$. Let us analyze $\cost(S')$. Suppose that the distance between $u$ and $v$ in $T_C$ is $br$ and such path contains $bl$ blocks. Then the number of edges w.r.t.\ $S$ grows by $|E(P_C)|$. The number of
credits w.r.t.\ $S$ decreases by at least $\frac{1}{4}br+bl+|E(P_C)|-1$ since we remove $br$ bridges, $bl$ blocks and $|E(P_C)|-1$ components (each one bringing at least one credit). However the number of credits also grows by $1$ since we create a new block $B'$ (which needs $1$ credit) or a new 2EC large component $C'$ (which needs $1$ additional credit w.r.t.\ the credit of $C$). Altogether $\cost(S)-\cost(S')\geq \frac{1}{4}br+bl-2$. We say that $P_C$ is cheap if the latter quantity is non-negative, and expensive otherwise. In particular $P_C$ is cheap if it involves at least $2$ block nodes or $1$ block node and at least $4$ bridges. Notice that a bridge-covering path, if at least one such path exists, can be computed in polynomial time. 

Before proving Lemma \ref{lem:bridgeCovering} we need the following two technical lemmas (see Figures \ref{fig:LBonSizeOfR(W)} and \ref{fig:BridgeCoveringPathsIntersection}). We say that a node $v\in V(T_C)\setminus \{u\}$ is reachable from $u\in V(T_C)$ if there exists a bridge-covering path between $v$ and $u$. Let $R(W)$ be the nodes in $V(T_C)\setminus W$ reachable from some node in $W\subseteq V(T_C)$, and let us use the shortcut $R(u)=R(\{u\})$ for $u\in V(T_C)$. Notice that $v\in R(u)$ iff $u\in R(v)$.

\begin{figure}
\begin{center}
\includegraphics[width=0.7\linewidth]{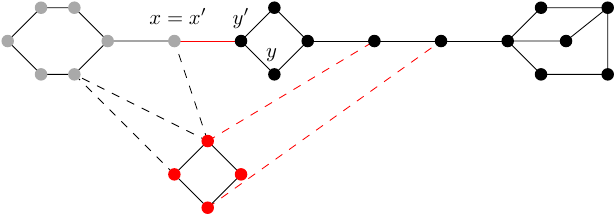}
\caption{Illustration of Lemma \ref{lem:reachable}. The solid edges define $S$ and the dashed ones are part of the remaining edges. The gray and black nodes induce the sets $X_C$ and $Y_C$, resp. The set $X$ is induced by the gray and red nodes. Notice that $y$ corresponds to the block with $4$ nodes (belonging to $C$). The red edges form the 3-matching used in the proof of the mentioned lemma.}
\label{fig:LBonSizeOfR(W)}
\end{center}
\end{figure}

\begin{lemma}\label{lem:reachable}
Let $e=xy\in E(T_C)$ and let $X_C$ and $Y_C$ be the nodes of the two trees obtained from $T_C$ after removing the edge $e$, where $x\in X_C$ and $y\in Y_C$. Then $R(X_C)$ contains a block node or $R(X_C)\setminus\{y\}$ contains at least $2$ lonely nodes.
\end{lemma}
\begin{proof}
Let us assume that $R(X_C)$ contains only lonely nodes, otherwise the claim holds. Let $X$ be the nodes in $G_C$ which are connected to $X_C$ after removing $Y_C$, and let $Y$ be the remaining nodes in $G_C$. Let $X'$ and $Y'$ be the nodes in $G$ corresponding to $X$ and $Y$, resp. Let $e'$ be the edge in $G$ that corresponds to $e$ and let $y'$ be the node of $Y'$ that is incident to $e'$. In particular if $y$ is a lonely node then $y'=y$ and otherwise $y'$ belongs to the block of $C$ corresponding to $y$.

Observe that both $X_C$ and $Y_C$ (hence $X$ and $Y$) contain at least one block node, hence, since $S$ is canonical, both $X'$ and $Y'$ contain at least 4 nodes. Therefore, we can apply the 3-matching Lemma \ref{lem:matchingOfSize3} to $X'$ and obtain a matching $M=\{u'_1v'_1,u'_2v'_2,u'_3v'_3\}$ between $X'$ and $Y'$, where $u'_i\in X'$ and $v'_i\in Y'$. Let $u_i$ and $v_i$ be the nodes in $G_C$ corresponding to $u'_i$ and $v'_i$, resp. We remark that $v_i\in Y_C$: indeed otherwise $v_i$ would be connected to $X_C$ in $G_C\setminus Y_C$, contradicting the definition of $X$. Let us show that the $v_i$'s are all distinct. Assume by contradiction that $v_i=v_j$ for $i\neq j$. Since $v'_i\neq v'_j$ (being $M$ a matching) and $v'_i,v'_j$ are both associated with $v_i$, this means that $v_i$ is a block node in $R(X_C)$, a contradiction.   

For each $u_i$ there exists a path $P_{w_iu_i}$ in $G_C\setminus E(T_C)$ between $u_i$ and some $w_i\in X_C$ (possibly $w_i=u_i$). Observe that $P_{w_iu_i}\cup u_iv_i$ is a bridge-covering path between $w_i\in X_C$ and $v_i\in Y_C$ unless $u_iv_i=xy$ (recall that the edges $E(T_C)$ cannot be used in a bridge-covering path). In particular, since the $v_i$'s are all distinct, at least two such paths are bridge-covering paths from $X_C$ to distinct (lonely) nodes of $Y_C\setminus \{y\}$, implying $|R(X_C)\setminus\{y\}|\geq 2$.
\end{proof}

Let $T_C(u,v)$ denote the path in $T_C$ between nodes $u$ and $v$. 
\begin{lemma}\label{lem:merge2paths}
Let $b$ and $b'$ be two block nodes of $T_C$ that are leaves. Let $u\in R(b)$ and $u'\in R(b')$ be nodes of $V(T_C)\setminus \{b,b'\}$. Suppose that $T_C(b,u)$ and $T_C(b',u')$ both contain some node $w$ (possibly $w=u=u'$) and $|E(T_{C}(b,u))\cup E(T_C(b,'u'))|\geq 4$. Then in polynomial time one can find a canonical 2-edge cover $S'$ such that $S'$ has fewer bridges than $S$ and $\cost(S')\leq \cost(S)$.  
\end{lemma}
\begin{proof}
Let $P_{bu}$ (resp., $P_{b'u'}$) be a bridge-covering path between $b$ and $u$ (resp., $b'$ and $u'$). Suppose that $P_{bu}$ and $P_{b'u'}$ share an internal node. Then there exists a (cheap) bridge-covering path between $b$ and $b'$ and we can find the desired $S'$ by the previous discussion. So next assume that $P_{bu}$ and $P_{b'u'}$ are internally node disjoint. Let $S'=S\cup E(P_{bu})\cup E(P_{b'u'})$. All the nodes and bridges induced by $E(P_{bu})\cup E(P_{b'u'})\cup E(T_C(b,u))\cup E(T_C(b',u'))$ become part of the same block or 2EC component $C'$ of $S'$. Furthermore $C'$ contains at least $4$ bridges of $C$ and the two blocks $B$ and $B'$ corresponding to $b$ and $b'$, resp. One has $|S'|=|S|+|E(P_{bu})|+|E(P_{b'u'})|$. Moreover $\credit(S')\leq \credit(S)+1-(|E(P_{bu})|-1)-(|E(P_{b'u'})|-1)-2-\frac{1}{4} \cdot 4$. In the latter inequality the $+1$ is due to the extra credit required by $C'$, the $-2$ to the removed blocks $B$ and $B'$ (i.e., blocks of $S$ which are not blocks of $S'$), and the $-\frac{1}{4} \cdot 4$ to the removed bridges. Altogether $\cost(S)-\cost(S')\geq \frac{1}{4}\cdot 4+2-3=0$. Notice that $S'$ has fewer bridges than $S$.
\end{proof}

\begin{figure}
\begin{center}
\includegraphics[width=0.6\linewidth]{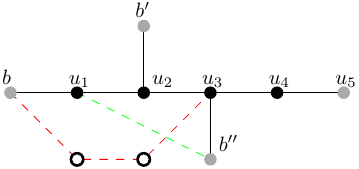}
\caption{The solid edges induce $T_C$. The white nodes correspond to contracted connected components other than $C$. The gray and black nodes are block and lonely nodes of $T_C$, resp. The red edges induce an (expensive) bridge-covering path covering the bridges $u_1u_2$, $u_2u_3$ and a bridge corresponding to $bu_1$. The green edge induces another (expensive) bridge-covering path. An edge $bb''$ or $bu_4$ would induce a cheap bridge-covering path. Lemma \ref{lem:merge2paths} can be applied with $u=u_3\in R(b)$ and $u'=u_1\in R(b')$ (both $T_C(b,u_3)$ and $T_C(b',u_1)$ contain $w=u_2$). In particular, one obtains $S'$ by adding to $S$ the red and green edges. Notice that the new component $C'$ of $S'$ which contains $V(C)$ has $4$ fewer bridges and $1$ fewer block than $C$.}
\label{fig:BridgeCoveringPathsIntersection}
\end{center}
\end{figure}

We are now ready to prove Lemma \ref{lem:bridgeCovering}.
\begin{proof}[Proof of Lemma \ref{lem:bridgeCovering}]
We will show how to find a canonical $2$-edge cover $S'$ of $G$ with $\cost(S')\leq\cost(S)$ and fewer bridges than~$S$. Applying this repeatedly implies the claim of the lemma. If there exists a cheap bridge-covering path $P_C$ (condition that we can check in polynomial time), we simply augment $S$ along $P_C$ hence obtaining the desired $S'$. Thus we next assume that no such path exists. 

Let $P'=bu_1\dots u_\ell$ be a longest path in $T_C$.  Notice that $b$ must be a leaf of $T_C$, hence a block node (corresponding to some block $B$ of $C$). Let us consider $R(b)$. Since by assumption there is no cheap bridge-covering path, $R(b)$ does not contain any block node. Hence $|R(b)\setminus\{u_1\}|\geq 2$ and $R(b)$ contains only lonely nodes by Lemma \ref{lem:reachable} (applied to $xy=bu_1$). 

We next distinguish a few subcases depending on $R(b)$. Let $V_i$, $i\geq 1$, be the nodes in $V(T_C)\setminus V(P')$ such that their path to $b$ in $T_C$ passes through $u_i$. Notice that $\{V(P'),V_1, \ldots,V_\ell\}$ is a partition of $V(T_C)$. We observe that any node in $V_i$ is at distance at most $i$ from $u_i$ in $T_C$ (otherwise $P'$ would not be a longest path in $T_C$). We also observe that as usual the leaves of $T_C$ in $V_i$ are block nodes. 
\begin{remark}\label{rem:V1V2}
This implies that all the nodes in $V_1$ are block nodes, and all the nodes in $V_2$ are block nodes or are lonely non-leaf nodes at distance $1$ from $u_2$ in $T_C$.
\end{remark}

\medskip\noindent{\bf (1)} There exists $u\in R(b)$ with $u\notin \{u_1,u_2,u_3\}\cup V_1\cup V_2$. By definition there exists a bridge-covering path between $b$ and $u$ containing at least $4$ bridges, hence cheap. This is excluded by the previous steps of the construction. 

\medskip\noindent{\bf (2)} There exists $u\in R(b)$ with $u\in V_1\cup V_2$. Since $u$ is not a block node, by Remark \ref{rem:V1V2} $u$ must be a lonely non-leaf node in $V_2$ at distance $1$ from $u_2$. Furthermore $V_2$ must contain at least one leaf block node $b'$ adjacent to $u$. Consider $R(b')$. By the assumption that there are no cheap bridge-covering paths and Lemma~\ref{lem:reachable} (applied to $xy=b'u$), $|R(b')\setminus \{u\}|\geq 2$. In particular $R(b')$ contains at least one lonely node $u'\notin \{b',b,u\}$. The tuple $(b,b',u,u')$ satisfies the conditions of Lemma \ref{lem:merge2paths} (in particular both $T_C(b,u)$ and $T_C(b',u')$ contain $w=u_2$), hence we can obtain the desired $S'$. 

\medskip\noindent{\bf (3)} $R(b)\setminus\{u_1\}=\{u_2,u_3\}$. Recall that $u_2$ and $u_3$ are lonely nodes. We distinguish $2$ subcases:

\medskip\noindent{\bf (3.a)} $V_1\cup V_2\neq \emptyset$. Take any leaf (block) node $b'\in V_1\cup V_2$, say $b'\in V_i$. Let $\ell'$ be the node adjacent to $b'$. By the assumption that there are no cheap bridge-covering paths and Lemma~\ref{lem:reachable} (applied to $xy=b'\ell'$), $R(b')\setminus \{\ell'\}$ has cardinality at least $2$ and contains only lonely nodes. Choose any $u'\in R(b')\setminus \{\ell'\}$. Notice that $u'\notin V_1$ by Remark \ref{rem:V1V2} (but it could be a lonely node in $V_2$ other than $\ell'$). 
The tuple $(b,b',u_3,u')$ satisfies the conditions of Lemma \ref{lem:merge2paths} (in particular both $T_C(b,u_3)$ and $T_C(b',u')$ contain $w=u_2$), hence we can compute the desired $S'$.

\medskip\noindent{\bf (3.b)} $V_1\cup V_2= \emptyset$. By Lemma \ref{lem:reachable} (applied to $xy=u_1u_2$) the set $R(\{b,u_1\})$ contains a block node or $R(\{b,u_1\})\setminus\{u_2\}$ contains at least 2 lonely nodes. Suppose first that $R(\{b,u_1\})$ contains a block node $b'$. Notice that $b'\notin R(b)$ by the assumption that there are no cheap bridge-covering paths, hence $u_1 \in R(b')$. Notice also that $b'\notin \{u_2,u_3\}$ since those are lonely nodes. Thus the tuple $(b,b',u_2,u_1)$ satisfies the conditions of Lemma \ref{lem:merge2paths} (in particular both $T_C(b,u_2)$ and $T_C(b',u_1)$ contain $w=u_2$), hence we can obtain the desired $S'$. 

The remaining case is that $R(\{b,u_1\})\setminus\{u_2\}$ contains at least $2$ lonely nodes. Let us choose $u'\in R(\{b,u_1\})\setminus\{u_2\}$ with $u'\neq u_3$. Let $P_{bu_2}$ (resp., $P_{u'u_1}$) be a bridge-covering path between $b$ and $u_2$ (resp., $u'$ and $u_1$). Notice that $P_{bu_2}$ and $P_{u'u_1}$ must be internally node disjoint otherwise $u'\in R(b)$, which is excluded since $R(b)\subseteq \{u_1,u_2,u_3\}$ by assumption. Consider $S'=S\cup E(P_{bu_2})\cup E(P_{u'u_1})\setminus \{u_1u_2\}$. Notice that $S'$ has fewer bridges than $S$. One has $|S'|=|S|+|E(P_{bu_2})|+|E(P_{u'u_1})|-1$, where the $-1$ comes from the edge $u_1u_2$. Furthermore $\credit(S')\leq \credit(S)+1-(|E(P_{bu_2})|-1)-(|E(P_{u'u_1})|-1)-\frac{1}{4} \cdot {4}-1$, where the $+1$ comes from the extra credit needed for the block or 2EC component $C'$ containing $V(B)$, the final $-1$ from the removed block $B$, and the $-\frac{1}{4} \cdot {4}$ from the at least $4$ bridges removed from $S$ (namely a bridge corresponding to $bu_1$, the edges $u_1u_2$ and $u_2u_3$, and one more bridge incident to $u'$). Altogether $\cost(S)-\cost(S')\geq -1 + 1 + 1 + 1 - 2 = 0$.
\end{proof}

\printbibliography

\end{document}